\documentclass[12pt]{article}     	 

\usepackage[margin=1in]{geometry} 
\usepackage{comment}                
\usepackage[tbtags]{amsmath}                
\usepackage{xcolor}                 
\usepackage{enumitem}               
\usepackage{amssymb}                
\usepackage{amsthm}                 
\usepackage{hyperref}               
\usepackage{graphicx}               
\usepackage{tikz, pgfplots}         
\usepackage{float}                  
\usepackage{natbib}                 
\usepackage{rotating}
\usepackage{setspace}
\usepackage{booktabs}
\usepackage{cleveref}
\usepackage{tcolorbox}
\usepackage{pdfpages}
\usepackage{bbm}
\usepackage[title]{appendix}

\theoremstyle{theorem}

\newtheorem{proposition}{Proposition}
\newtheorem{corollary}{Corollary}
\newtheorem{lemma}{Lemma}

\DeclareMathOperator*{\argmax}{arg\,max}
\DeclareMathOperator*{\argmin}{arg\,min}
\numberwithin{equation}{section}

\title{%
 Nonlinear Impulse Response Functions and Local Projections\thanks{We thank Burda, M., Lu, Y., Melino, A., Renault, E., the participants of the 2023 NBER/NSF Time Series Conference (Montreal), and of the 38th Canadian Econometrics Study Group Annual Meeting (Hamilton) for their helpful comments. A previous version of this paper was circulated on ArXiv: https://arxiv.org/abs/2305.18145.}}
\author{Gouri\'eroux, C.,\footnote{University of Toronto, Toulouse School of Economics and CREST, email: \textit{Christian.Gourieroux@ensae.fr}} and Q., Lee\footnote{University of Toronto, email: \textit{qt.lee@mail.utoronto.ca}}}
\date{Revised, \today}
\begin{document}
		\setstretch{1}
	\maketitle
	\begin{abstract}
		\noindent The goal of this paper is to extend the nonparametric estimation of Impulse Response Functions (IRF) by means of local projections in the nonlinear dynamic framework. We discuss the existence of a nonlinear autoregressive representation for Markov processes and explain how their IRFs are directly linked to the Nonlinear Local Projection (NLP), as in the case for the linear setting. We present a fully nonparametric LP estimator in the one dimensional nonlinear framework, compare its asymptotic properties to that of IRFs implied by the nonlinear autoregressive model and show that the two approaches are asymptotically equivalent. This extends the well-known result in the linear autoregressive model by Plagborg-Moller and Wolf (2017). We also consider extensions to the multivariate framework through the lens of semiparametric models, and demonstrate that the indirect approach by the NLP is less accurate than the direct estimation approach of the IRF.	   \\
		
		\vspace{0in}
	\end{abstract}

	\newpage
	
	\section{Introduction}
	
The Impulse Response Function (IRF), a notion initially introduced to economics by Frisch (1933), is a popular tool for macroeconomists to study the effects of shocks in a dynamic framework. In the current literature, there are two popular methods in which the IRF can be estimated. The first is called the direct approach which requires the practitioner to specify an economic model in the form of a (linear) Structural Vector Autoregression (SVAR) [see Sims (1980), Pesaran and Shin (1998), Christiano (2012), Ramey (2016), Kilian and Lutkepohl (2017) for examples]. The IRF is then estimated by simulating a perturbed and a baseline path implied by the autoregressive model, and computing the average difference. The second, which has garnered significant interest in recent studies, is an indirect nonparametric approach by means of a local linear projection [see Jorda (2005) for the definition of local linear projection\footnote{Also see Dufour and Renault (1998) for an interpretation in terms of short and long run causality. See their discussion of Impulse Response Functions on page 1113.}, Chang and Sakata (2007) for its derivation from long run regressions]. Each method boasts its own benefits [e.g. efficiency gains in the direct approach, robustness to misspecification in the indirect approach] and are directly linked to one another in the linear dynamic framework [Plagborg-Moller and Wolf (2021)]. \\
	

However, linear methods are rather restrictive and constitute a major source of misspecification risk. This arises due to the omission of nonlinear dynamics - for instance, ``linear models cannot adequately capture asymmetries that may exist in business cycle fluctuations" [Koop et al. (1996)]. Indeed, nonlinear behaviour has continued to feature in a number of applied settings, including but not limited to: [1] Cyclical behaviour due to a tent map or logistic effects [Frank and Stengos (1988)]. [2] Speculative bubbles and local trends that are captured by causal/noncausal models [Lanne and Saikkonen (2011), Jorda et al. (2015), Gouri\'eroux and Jasiak (2017), (2022), Bec et al. (2020)]. [3] Jumps and disasters [Christoffersen, Du and Elkawhi (2017), Wang (2019), Paul (2020), Gouri\'eroux, Monfort, Mouabbi and Renne (2021)]. [4] Regime switches, introduced to model asymmetric responses of output to oil prices for instance [Hamilton (2003), (2011), Kilian and Vigfusson (2011), (2017)]. [5] Modern macroeconomic models, which need to account for transition to low carbon economics [Metcalf and Stock (2020)], network effects due to supply chain failures [Kuiper and Lansink (2013)], fast technological changes arising from digitalization, potential epidemiological occurrences [Toda (2020)], asymmetric inflation expectations [Baqace (2020)], asymmetric effects of weather shocks [Cashin et al. (2017), DeTruchis et al. (2024)], or for nonlinear summaries of income inequality [Frost and Van Stralen (2018)]\footnote{Such phenomena cannot be captured by simply log-linearizing nonlinear dynamic models around an equilibrium value. In particular, ad hoc methods tend to linearize the first-order conditions of the intertemporal optimization problem, leading to linear rational expectation models. Among the multiplicity of solutions, only stationary fundamental solutions with linear dyanmics are selected, whereas other stationary fundamental solutions with speculative bubbles are rejected ex-ante [see Gourieroux, Jasiak and Monfort (2020) for the complete set of stationary solutions.]. }. \\

The econometric literature on nonlinear IRFs is rapidly developing, with a number of definitions that have been proposed and discussed [Gallant et al. (1993), Koop, Pesaran and Potter (1996), Pesaran and Shin (1998), Gouri\'eroux and Jasiak (2005), Goncalves et al. (2024a)]. Among these studies, it is widely accepted that a nonlinear IRF is characterized by a perturbation of the structural innovation at horizon 1, that is $\varepsilon_{t+1}+\delta$\footnote{See Lee (2025) for a discussion of the various definitions of the nonlinear IRF, and why only this current definition can be linked to nonlinear Forecast Error Variance Decompositions (FEVD).}. While estimation of nonlinear IRFs have been studied extensively in the nonlinear autoregressive framework, the development of estimation methods via nonlinear local projections in this context is rather limited. Moreover, it is an open question whether the link between the direct and indirect methods remains intact in the nonlinear framework. This apparent gap in the literature is what motivates the discussion in this paper. \\

We introduce the nonlinear autoregressive representation of a Markov process and assume that the practitioner has applied restrictions to ensure the identifiability of its structural innovations\footnote{ It is known that without additional parametric or semi-nonparametric restrictions, this representation is not unique as well as the associated nonlinear Gaussian innovation [Gourieroux and Lee (2025)]}. We then establish the link between the nonlinear IRFs implied by this autoregressive model, and the IRFs obtained by means of a nonlinear local projection. Furthermore, we propose nonparametric estimators for both approaches and demonstrate their asymptotic equivalence in the univariate nonlinear setting. For the multivariate setting, the curse of dimensionality severely inhibits the extension of our estimation results. Nonetheless, we also propose semi-parametric alternatives which allow us to mitigate some of the issues related to dimensionality. \\

This paper is organized as follows. Section 2 introduces the nonlinear autoregressive reprsentation. Section 3 provides definitions of shocks, nonlinear IRFs and the nonlinear local projection. We then propose nonparametric estimators of the IRF in Section 4 and discuss their asymptotic properties. The multivariate extension is considered in Section 5, and Section 6 concludes. All derivations of statistical properties of nonparametric and/or parametric estimators of nonlinear IRF are provided in the appendices and standard asymptotic properties of kernel estimators are reviewed in the online appendices.

	\section{Nonlinear Autoregressive Representation of a Markov Process}
	
First we introduce the nonlinear autoregressive process and discuss how to obtain its future trajectories. Then, we provide some simple examples of univariate Markov processes that fall under this framework. 

	\subsection{The Properties}
	
Let us consider an $n$-dimensional Markov process of order 1, denoted $y = (y_t)$, with values in  $\mathcal{Y}=\mathbb{R}^n$. If this process has a continuous distribution, the Markov condition can be written on its transition density as $f(y_t|\underline{y_{t-1}})=f(y_t|y_{t-1})$, where $\underline{y_{t-1}}=(y_{t-1},y_{t-2},...)$. There is an equivalent way of writing the Markov condition.
\begin{proposition}
$(y_t)$ is a Markov process of order 1 on $\mathcal{Y}=\mathbb{R}^n$ with a strictly positive transition density: $f(y_t|y_{t-1})>0$, $\forall y_t,y_{t-1}$, if and only if it admits a nonlinear autoregressive representation:
\begin{equation}\label{nlar}
	y_t = g(y_{t-1};\varepsilon_t), \ t \geq 1,
\end{equation} 
where the $\varepsilon_t$'s are independent and identically distributed $N(0,Id)$ variables, with $\varepsilon_t$ being independent of $y_{t-1}$, and $g$ is a one-to-one transformation with respect to $\varepsilon_t$, that is continuously differentiable with a strictly positive Jacobian. The process $(\varepsilon_t)$ defines a Gaussian nonlinear innovation of the process $(y_t)$. 
\end{proposition} 

\textbf{Proof:} This is directly deduced from Rosenblatt (1952) [see also Gourieroux and Jasiak (2005) for the terminology of Gaussian nonlinear innovation]. Q.E.D. \\

The condition of Gaussianity on the nonlinear innovation is just a normalization condition. Then, the nonlinear dynamic features can be introduced since $y_t$ is a nonlinear function of $y_{t-1}$ for a given $\varepsilon_t$, and/or a nonlinear function of $\varepsilon_t$ for a given $y_{t-1}$, and/or by nonlinear cross-effects of $y_{t-1}$ and $\varepsilon_{t}$. In particular this nonlinear representation exists even if $y_t$ has marginal and/or conditional fat tails. \\

\textbf{Remark 1:} It is known that, without additional parametric or semi-nonparametric restrictions on the dynamic Markov process, the nonlinear autoregressive representation, i.e. function $g$ and Gaussian nonlinear innovations, are not defined in a unique way [see the literature on nonlinear independent component analysis (ICA) [Comon (1994), Hyvarinen and Pajunen (1999), Hyvarinen et al. (2019), and Gourieroux and Lee (2025) for identification of nonlinear autoregressive representation.]. We assume in the rest of the paper that such restrictions are introduced, that will imply the identification of the structural shocks and of the associated nonlinear IRF.  \\ 

\textbf{Remark 2:} Proposition 1 is easily extended to a multidimensional Markov process of order $p$. The nonlinear autoregressive representation becomes:
	\begin{equation}\label{NLAR}
	y_t = g(y_{t-1},...,y_{t-p};\varepsilon_t).
\end{equation}
We consider the case $p=1$ for ease of exposition and to avoid the curse of dimensionality in nonparametric analysis.\\
%

\textbf{Remark 3:} The literature has suggested an alternative definition of nonlinear innovation (in the one dimensional framework) as:
\begin{equation*}
\varepsilon_t^* = \frac{y_t-\mathbb{E}(y_t|y_{t-1})}{\sqrt{\mathbb{V}(y_t|y_{t-1})}},
\end{equation*}
[see Blanchard and Quah (1989), Koop, Pesaran and Potter (1996), eq. (1), where $\varepsilon_t^*$ is denoted $v_t$.]. It is easily checked that $\varepsilon_t^*$ is not independent of $y_{t-1}$ in general and thus not appropriate for shocking $\varepsilon_t^*$ in the construction of nonlinear IRFs. \\

	\subsection{Future Trajectories}
	The nonlinear autoregressive representation \eqref{nlar} is appropriate for describing and/or simulating the future values of the process. We proceed by recursive substitution as follows. First, we have: 
	\begin{equation*}
		\begin{split}
			y_{t+1} & = g(y_{t};\varepsilon_{t+1}), \\
			 y_{t+2} & = g(y_{t+1};\varepsilon_{t+2}), 	 \\ 
			\vdots & \\
			 	y_{t+h}  & = g(y_{t+h-1},\varepsilon_{t+h}).
		\end{split}
	\end{equation*}
By substituting equation $y_{t+1}$ into $y_{t+2}$, we obtain:
\begin{equation*}
	 y_{t+2}  = g(y_{t+1};\varepsilon_{t+2}) =g[g(y_t;\varepsilon_{t+1});\varepsilon_{t+2}] \equiv g^{(2)}(y_t;\varepsilon_{t+1},\varepsilon_{t+2}).
\end{equation*}
We can perform these substitutions repeatedly for $y_{t+j}$, $j=1,...,h$ to obtain:
\begin{equation}
	y_{t+h}=g^{(h)}(y_t;\varepsilon_{t+1:t+h}),
\end{equation}
where $\varepsilon_{t+1:t+h}=(\varepsilon_{t+1},...,\varepsilon_{t+h}),h\geq1$. In particular we have the recursive formula:
\begin{equation}
	y_{t+h} = g^{(h)}(y_t;\varepsilon_{t+1:t+h})=g^{(h-1)}[g(y_t;\varepsilon_{t+1});\varepsilon_{t+2:t+h}].
\end{equation}

\subsection{Examples}

\textbf{Example 1: Conditionally Gaussian Model}\\

 In the one dimensional case, it includes the Double Autoregressive (DAR) model of order one [Weiss (1984), Borkovec and Kluppelberg (2001), Ling (2007)], introduced to account for conditional heteroscedasticity and given by:
\begin{equation}
	y_t = \gamma y_{t-1} + \sqrt{\alpha + \beta y_{t-1}^2} \ \varepsilon_t, \alpha >0,\beta\geq0,
\end{equation}
where $\varepsilon_t$ is $IIN(0,1)$. The DAR model has a strictly stationary solution if the Lyapunov coefficient is negative: 
\begin{equation}
 \mathbb{E}\log |\gamma + \sqrt{\beta} \ \varepsilon|<0.
\end{equation}
The conditional heteroscedasticity can create fat tails for the stationary distribution of process $(y_t)$. This process has second-order moments if moreover:
\begin{equation*}
	\gamma^2 + \beta < 1.
\end{equation*}
When $	\gamma^2 + \beta > 1$ and $\mathbb{E}\log |\gamma + \sqrt{\beta} \ \varepsilon|Z<0$, the process $(y_t)$ is strictly stationary with infinite marginal variance (but finite conditional variance). \\

\textbf{Example 2: Time Discretized Diffusion Process}\\

The discrete time process is $y_t=y(t),t=1,2,...$, where the underlying process $y(\tau)$ is defined in continuous time $\tau \in (0,\infty)$ by a (multivariate) diffusion equation:

\begin{equation}\label{diffusion}
	dy(\tau)=m[y(\tau)]d\tau + D[y(\tau)]dW(\tau),
\end{equation}
where $W$ is a Brownian motion with $\mathbb{V}[dW(\tau)]=d\tau$. 
The diffusion equation \eqref{diffusion} is the analogue of the conditionally Gaussian model written in an infinitesimal time unit. The class of time discretized diffusion contains the Gaussian AR(1) process, the autoregressive gamma (ARG) process, that is the time discretized Cox, Ingersoll, Ross process [Cox, Ingersoll and Ross (1985)], or the time discretized Jacobi process [Karlin and Taylor (1981), Gour\'ieroux and Jasiak (2009)], with values in $(-\infty,+\infty),(0,+\infty),(0,1)$, respectively. This is easily extended to the multivariate framework, including stochastic volatility models  [Hull and White (1987), Heston (1993)] and the Wishart Autoregressive models (WAR) [Gour\'ieroux, Jasiak and Sufana	 (2009)].

\section{Shocks, Impulse Response Functions and Local Projections}

The analysis of shocks and of their propagations in a nonlinear stochastic system, i.e. the impulse response functions, can serve different objectives: [i] It can be a technical tool to analyze the dynamic properties of the system and their robustness. [ii] It can be used as a counterfactual experiment: ``\textit{What would have arisen if...?}". [iii] In a more structural approach, they can be used to evaluate ex-ante the effects of a new policy. This latter objective is more structural and demands to define what can be controlled and at what magnitude. The literature agrees on the fact that a shock performed at date $t+1$ cannot have impacts on the past. It has to be written on an innovation that is on a (multidimensional) variable, independent of its past. It also agrees on the observation that these innovations are not necessarily uniquely defined and that this identification issue has to be taken into account. However, the definitions of the shocked innovations can vary.

\subsection{Shocks}
We consider a transitory shock of magnitude $\delta \in \mathbb{R}^n$ hitting the innovation $\varepsilon_{t+1}$. This is a system wide shock and not a variable specific shock. Thus, $\varepsilon_{t+1}$ is replaced by $\varepsilon_{t+1}+\delta$, while the past values and other future innovations are unchanged. Then we can compare the future trajectories before this shock: 
\begin{equation*}
y_{t+h} = g^{(h-1)}[g(y_t;\varepsilon_{t+1});\varepsilon_{t+2:t+h}], h\geq 1,
\end{equation*}
and the future trajectories after the shock is applied:
\begin{equation*}
y^{(\delta)}_{t+h} = g^{(h-1)}[g(y_t;\varepsilon_{t+1}+\delta);\varepsilon_{t+2:t+h}], h\geq 1. 
\end{equation*}
These trajectories correspond to a conceptual experiment, since the innovations do not necessarily have economic interpretations. Thus, we trace out the effects of the shock on the outcome future values. 

\subsection{Impulse Response Function (IRF)}
Both the initial trajectory and the shocked trajectory are stochastic since they depend on the unknown future values of the innovations $\varepsilon_{t+1},...,\varepsilon_{t+h}$. It is usual to replace the comparison of the stochastic trajectories by a comparison of their summaries. In a nonlinear dynamic framework, a drift of $\delta$ on $\varepsilon_{t+1}$ can have very complicated effects on the distribution of the process. To get enough information on these effects, several notions of Impulse Response Functions (IRFs) will have to be introduced. While our focus here is on the one dimensional case $n=1$, the extensions to the multidimensional framework are straightforward. In particular, it is known that the joint distribution of the components of $y_t$ is characterized by the knowledge of the one-dimensional distributions of the linear combinations $a'y_t$, with $a$ varying. Then, in the multidimensional framework, all the IRF's below can be applied to such linear combinations. \\

In the definitions below, $\mathbb{E}[\cdot|y_t]$ denotes a conditional expectation, that is the best approximation by a square integrable nonlinear function of $y_t$. In nonlinear dynamic models, it is not equal to the theoretical linear regression, that is the projection on the linear space generated by $y_t$. Thus, the distinction between conditional expectations and projection by linear regression is crucial in nonlinear dynamic frameworks \footnote{It also does not coincide with the linear regression on quadratic or cubic functions of $y_t$ [see Jorda (2005), Section II, for this suggestion, called ``flexible local projection"]. To insist on this difference, Billingsley (1986) proposed a double bar convention for the conditional expectation $\mathbb{E}[y_{t+h}||y_t]$. We use a single bar later on.}. This conditioning on the past is especially important in nonlinear dynamic models where the effects of a shock (i.e. the IRF) depends on both the current environment (Are we close to a tipping point?) and the magnitude of the shock (Is the shock sufficiently large to cross this tipping point?)\footnote{In this respect, an unconditional definition of IRF is not appropriate and can lead to misleading interpretations of the IRF [see Goncalves et al. (2021), Definition 1, or Goncalves et al. (2024a,b)].}. \\ 

\textbf{(i) IRF for pointwise prediction} \\ 

It is defined by: 
\begin{equation}\label{IRF_pointwise}
	IRF(h,\delta|y_t)=\mathbb{E}[y_{t+h}^{(\delta)}-y_{t+h}|y_{t}].
\end{equation}

This is a functional parameter that depends on both the horizon h and magnitude $\delta$ of the shock. It also depends on the history by means of $y_t$. In particular, the IRF has to be updated at each new observation. It is called the Generalized IRF in Koop, Pesaran, and Potter (1996)\footnote{It is also called expected IRF in some other literatures, to highlight the pointwise predictor by conditional expectations.}. \\

\textbf{(ii) IRF for pointwise prediction of the transformed $y$} \\ 

Let us consider a nonlinear transformation $a(y)$ of $y$. The IRF is defined by: 
\begin{equation}
	IRF(h,\delta,a|y_t)=\mathbb{E}[a(y_{t+h}^{(\delta)})-a(y_{t+h})|y_{t}].
\end{equation}
For instance, if $n=1$ and $a_y(y_{t+h})=\mathbbm{1}_{y_{t+h}<y}$, we get: 
\begin{equation}
IRF(h,\delta,a_y|y_t) = \mathbb{P}[y^{(\delta)}_{t+h}<y|y_t]-\mathbb{P}[y_{t+h}<y|y_t], \ \text{for any} \ h,\delta,y,
\end{equation}
and the possibility to compare the predictive distributions at all horizons. By inverting these cumulative distribution functions, we can also consider the effect of shocks on the conditional quantiles, that are the conditional Value-at-Risk (VaR) [Gouri\'eroux and Jasiak (2005)]. These IRF based on VaR are used to evaluate the sensitivity of reserves for banks in prudential supervision. \\

\textbf{(iii) IRF for dynamic features} \\ 

It is also important to evaluate the effects on dynamic features, such as the conditional serial dependence at lag 1. This would lead to the IRF's of the type: 
\begin{equation}
	IRF(h,\delta|y_t)=\mathbb{E}[y_{t+h}^{(\delta)}y_{t+h-1}^{(\delta)}-y_{t+h}y_{t+h-1}|y_{t}].
\end{equation}
This can be applied to the autocorrelation function (ACF) as well as the squared ACF. Indeed, in nonlinear dynamic models, a small shock can significantly change the dynamics, including the linear ACF due to chaotic features. \\

\textbf{(iv) Joint IRF} \\ 

All IRF's above are computed without taking into account the dependence between the trajectories. Other tranformations can reveal this cross-dependence, such as: 
\begin{equation}
	JIRF(h,\delta|y_t)=\mathbb{E}[(y_{t+h}^{(\delta)}-y_{t+h})^2|y_{t}],
\end{equation}
that includes a cross term: $\mathbb{E}[y_{t+h}^{(\delta)}y_{t+h}|y_t]$. 

\subsection{ Nonlinear Local Projections (NLP)}

Let us consider the $IRF(h,\delta)$ defined in \eqref{IRF_pointwise}, and introduce the pointwise prediction at horizon $h$:
\begin{equation}\label{mh_lp}
	m^{(h)}(y_t) = \mathbb{E}(y_{t+h}|y_t).
\end{equation}
By the Markov property, this pointwise prediction depends on the past by means of $y_t$ only. Such direct prediction has been called local projection in the linear dynamic framework [see Jorda (2005)]. By analogy with what is known in this linear framework [Plagborg-Moller and Wolf (2021), Montiel Olea and Plagborg-Moller (2021),(2022)], we will relate the IRF's and the NLP associated with prediction \eqref{mh_lp}.

\begin{proposition}
\begin{equation}\label{prop3}
\begin{split}
	IRF(h,\delta|y_t)&=\mathbb{E}\{m^{(h-1)}[g(y_t;\varepsilon_{t+1}+\delta)]-m^{(h-1)}[g(y_t;\varepsilon_{t+1})]|y_t\} \\ 
	& = \int \{m^{(h-1)}[g(y_t;\varepsilon+\delta)]-m^{(h-1)}[g(y_t;\varepsilon)]\}\phi(\varepsilon)d\varepsilon,\\ 
\end{split}
\end{equation}
where $\phi$ is the density of $N(0,Id)$, for $h\geq 1$\footnote{For $h=1$, $m^{(0)}$ is the identity function.}. 
\end{proposition}

\textbf{Proof:} We have:
\begin{equation*}
\begin{split}
& \mathbb{E}[y_{t+h}^{(\delta)}-y_{t+h}|y_t]\\
= & \mathbb{E}\{g^{(h-1)}[g(y_t;\varepsilon_{t+1}+\delta);\varepsilon_{t+2:t+h}]-g^{(h-1)}[g(y_t;\varepsilon_{t+1});\varepsilon_{t+2:t+h}]\} \\
= & \mathbb{E}\{\mathbb{E}\{g^{(h-1)}[g(y_t;\varepsilon_{t+1}+\delta);\varepsilon_{t+2:t+h}]-g^{(h-1)}[g(y_t;\varepsilon_{t+1});\varepsilon_{t+2:t+h}]|y_t,\varepsilon_{t+2:t+h}\}|y_t\}\\
& \text{(By the Law of Iterated Expectation)} \\
=&  \mathbb{E}[m^{(h-1)}[g(y_t;\varepsilon+\delta)]-m^{(h-1)}[g(y_t;\varepsilon)]|y_t],
\end{split}
\end{equation*}
by definition of $m^{(h-1)}$. The result follows. Q.E.D. \\ 

Proposition 3 is easily extended to the other types of IRF's. The right hand side of equation \eqref{prop3} defines the NLP interpretation of the IRF. As seen below, it differs from the standard formula since the effect of the shock is nonlinear in $\delta$ and has to be integrated out with respect to $\varepsilon$. The integration in equation \eqref{prop3} can be avoided in special cases. 

\begin{corollary}
	Let us assume that:
	\begin{equation*}
		m^{(h-1)}[g(y_t,\varepsilon_{t+1})] = a^{(h-1)}(y_t)\varepsilon_{t+1}+b^{(h-1)}(y_t).
	\end{equation*}
Then:
\begin{equation*}
	IRF(h,\delta|y_t) = a^{(h-1)}(y_t)\delta.
\end{equation*}
\end{corollary}

The condition in Corollary 1 is satisfied in the linear dynamic models usually considered in the literature. It provides the same IRF's as the IRF comparing two shocked trajectories, one with $\varepsilon_{t+1}$ replaced by $\delta$, and another with $\varepsilon_{t+1}$ replaced by $0$\footnote{This corresponds to the so-called MIT definition of IRF. However, this definition is not appropriate for more complicated nonlinear features as noted in Kolesar and Plagborg-Moller (2024).}. Indeed, only the difference matters. In this case it can also be normalized by focusing on $\delta=1$, since $IRF(h,\delta|y_t)$ becomes linear in $\delta$. However, without the strong linearity restriction in Corollary 1, the IRF will depend on the current environment $y_t$, and on the magnitude of the multivariate shock $\delta$, in a nonlinear way (in general). Other simplifcations of formula \eqref{prop3} are expected in specific dynamic models as fully recursive structural models [see Goncalves et al. (2021), Section 5 with the i.i.d. assumption on the ``regressors" and the discussion in Section 5.3]. As mentioned by the authors this assumption of full recursivity is not economically plausible in general.

\subsection{Sequence of Shocks}

We have considered in the subsections above the case of a transitory (i.e. isolated) shock performed at date $t+1$. In a linear dynamic framework, it is standard to also consider sequence of shocks and in particular permanent shocks applied after this date. However, the analysis of the IRF following a sequence of shocks is significantly different in a nonlinear dynamic framework. For illustration, let us consider below two consecutive shocks of magnitude $\delta_1$ and $\delta_2$ at dates $t+1$ and $t+2$ respectively. Then, the future trajectories are:
\begin{equation}
	y_{t+h}^{(\delta_1,\delta_2)} = g^{(h-2)}[g^{(2)}(y_t;\varepsilon_{t+1}+\delta_1,\varepsilon_{t+2}+\delta_2);\varepsilon_{t+3:t+h}], \ h \geq 2,
\end{equation}
and the IRF becomes:
\begin{equation}
	IRF(h,\delta_1,\delta_2|y_t) = \mathbb{E}\left[y_{t+h}^{(\delta_1,\delta_2)}-y_{t+h}|y_t\right].
\end{equation}
The IRF can be written in terms of nonlinear local projections as:
\begin{equation}
	\begin{split}
IRF(h,\delta_1,\delta_2|y_t) & = \int\int\left\{m^{(h-2)}\left[g^{(2)}(y_t;\varepsilon_{t+1}+\delta_1,\varepsilon_{t+2}+\delta_2)\right]\right.\\
& - \left.m^{(h-2)}\left[g^{(2)}(y_t;\varepsilon_{t+1},\varepsilon_{t+2})\right]\right\}\phi(\varepsilon_1)\phi(\varepsilon_2)d\varepsilon_1 d\varepsilon_2, \ h \geq 2.
	\end{split}
\end{equation}
This expression extends formula \eqref{prop3} in Proposition 3 to a sequence of two consecutive shocks [see Diercks et al. (2023) for another attempt to extend the local projection to this framework]. In the linear dynamic framework, this IRF can be decomposed as:
\begin{equation}\label{3.11}
	IRF(h,\delta_1,\delta_2|y_t) = 	IRF(h,\delta_1,0|y_t) + 	IRF(h-1,0,\delta_2|y_t) ,
\end{equation}
that is directly deduced from the IRF's of the isolated shocks. This equality \eqref{3.11} is no longer valid in the nonlinear dynamic framework. In general, the $IRF(h;\delta_1,\delta_2)$ is not a deterministic function of the IRF's associated with the isolated shocks and, morever, the combined effect of the two shocks can have diminishing, neutral, or amplifying effects on the values $y_{t+h}$, due to cascading effects. Therefore ``with nonlinear models considering the IRF that measures the effect of a shock of a given size hitting at a given period can be very misleading" [Koop (1996), p 136]. Instead of computing the IRF for all values of the pair $(\delta_1,\delta_2)$, it has also been proposed to draw them in some distribution in order to account for the uncertainty on $(\delta_1,\delta_2)$ [Koop (1996)]. \\

Note that the definition and expansion of IRF are written directly from the autoregressive representation without deriving the nonlinear moving average representation of the process [see the discussion in Adamek et al. (2024), p. 324, for the drawback of such an approach even in the linear dynamic framework]. Indeed this approach based on autoregressive representation is more appropriate for deriving by simulation the different IRFs.

\section{Nonparametric Inference for n=1}

An argument for local projection analysis of IRF's is that this is a nonparametric approach that is less sensitive to possible mispecification of the lag in the autoregressive parametric dynamic. This argument has been criticized [see Kilian and Kim (2009), p.1466]. Moreover, it is given assuming a linear dynamic model, whereas one of the main sources of misspecification is likely the omitted nonlinear dynamics\footnote{The omitted nonlinearities do not only concern omitted conditional heteroscedasticity [see e.g. Kilian and Kim (2011), Herbst and Johannsen (2021), Montiel-Olea and Plagborg Moller (2021) for bias correction of the IRF standard errors under conditional heteroscedasticity either by bootstrap, or by expansions.]}. Let us now compare the direct IRF estimation approach and the LP estimation approach in a nonlinear dynamic framework. Due to the curse of dimensionality of conditional nonparametric analysis, we consider the one dimensional case $n=1$ and the observations $y_1,...,y_T$ of the process. By the Markov property, we also assume the same lag equal to 1 for the direct and NLP estimation approaches. Then the two approaches will agree and nevertheless are both nonparametric. 

\subsection{Two Nonparametric Estimation Approaches}

When $n=1$, the nonlinear AR(1) [NLAR(1)] model is identifiable (up to a change of sign on $\varepsilon_t$). Then we can consider a nonparametric functional estimator of function $g$ as:
\begin{equation}
	\hat{g}_T(y_{t-1};\varepsilon)=\hat{Q}_T[\Phi(\varepsilon)|y_{t-1}],
\end{equation}
based on a kernel estimation of the conditional quantile of $y_t$ given $y_{t-1}$:
\begin{equation}\label{kapp}
	\hat{Q}_T(\alpha|y) = \argmin_{q} \sum^{T}_{t=1}K\left(\frac{y_{t-1}-y}{b_T}\right)\{\alpha(y_t-q)^+ + (1-\alpha)(y_t-q)^-\},
\end{equation}
where $K$ is the kernel, $b_T$ the bandwidth and $\alpha$ the critical level\footnote{The kernel and bandwidth could depend on the level $\alpha$ to treat differently the smoothing in the standard values and in the tails. This question is out of the scope in the present paper. }. Then, this estimated NLAR(1) model can be used to simulate future trajectories $y^s_{t+k}$, $k=1,..,h$, $s=1,2,...,S$, by applying recursively the nonlinear autoregressive equation:

\begin{equation}
\hat{y}^s_{t+k}=\hat{g}_T(\hat{y}^s_{t+k-1};\varepsilon^s_{t+k}), k=1,...,h,
\end{equation} 
where the $\varepsilon_t$'s are independently drawn from the standard normal distribution, with starting value $\hat{y}^s_t=y_t$ \footnote{These simulated values depend on $T$ by means of $\hat{g}_T$. This index $T$ is omitted below for exposition, and they also depend on $y_t$.}. For illustration we focus on the IRF for $\mathbb{E}(y_{t+h}|y_t)$. \\

\textbf{(i) Direct estimation of} $IRF(h,\delta|y_t)$\\ 

The direct approach follows the steps below: \\

Step 1: Simulate $\hat{y}_{t+h}^s$ and $\hat{y}_{t+h}^{(\delta),s}$, $s=1,...,S$, based on $\hat{g_T}$. \\

Step 2: Compute: 
\begin{equation}\label{dd_irf}
\widehat{IRF}_T(h,\delta|y_t) = \frac{1}{S}\sum_{s=1}^{S}(\hat{y}_{t+h}^{(\delta),s}-\hat{y}_{t+h}^s).
\end{equation}

\textbf{(ii) Estimation by means of Local Projection (Indirect Estimation)}\\ 

The indirect approach is based on formula \eqref{prop3} in Proposition 3. The steps become: \\

Step 1: Compute the Nadaraya-Watson estimate\footnote{This estimator is based on a smoothing with respect to the conditioning value $y$. This differs from the smooth local projections in which the smoothing is with respect to $h$ [Barnichon and Brownlees (2019), Plagborg-Moller and Wolf (2021)].} of the function $m^{(h)}(\cdot)$. This provides:
\begin{equation}
\begin{split}
	\hat{m}_{T,h}(y) & = \argmin_{m} \sum^{T-h}_{t=1}K\left(\frac{y_t-y}{b_T}\right)(y_{t+h}-m)^2\\
	& = \sum_{t=1}^{T-h}\left[K\left(\frac{y_t-y}{b_T}\right)y_{t+h}\right]/\sum_{t=1}^{T-h}K\left(\frac{y_t-y}{b_T}\right),\\
\end{split}
\end{equation}
where $K$ is a kernel and $b_T$ a bandwidth. \\

Step 2: Simulate $\varepsilon^s_{t+1}$, $\hat{y}_{t+1}^s$ and $\hat{y}_{t+1}^{(\delta),s}$, $s=1,...,S$.  \\ 

Step 3: Compute: 
\begin{equation}\label{lp_irf}
\widehat{\widehat{IRF}}_T(h,\delta|y_t) = \frac{1}{S}\sum_{s=1}^{S}[\hat{m}_T^{(h-1)}(\hat{y}_{t+1}^{(\delta),s})-\hat{m}_T^{(h-1)}(\hat{y}_{t+1}^s)].
\end{equation}

We will compare, in the next subsection, the asymptotic properties of these functional estimators given in \eqref{dd_irf} and \eqref{lp_irf}. But clearly, to get $IRF(h,\delta|y_t)$, $h=1,...,H$ for a given $\delta$, the first approach requires the nonparametric estimation of function $g$ at $2HS$ values (that are the $\hat{y}^s_{t+k},\hat{y}_{t+k}^{(\delta),s},k=1,...,h, s=1,...,S)$. The second approach requires the nonparametric estimation of function $g$ at $2S$ values (that are the $\hat{y}^s_{t+1},\hat{y}_{t+1}^{(\delta),s},s=1,...,S)$, and of functions $\hat{m}^{(h)}$ at $2(h-1)S$ values. Therefore, they seem numerically equivalent in the nonlinear dynamic framework\footnote{Note that the direct and indirect estimators coincide for $h=1$.}. \\ 

Their asymptotic accuracies will depend on the nonparametric estimation method that is used and in particular on the choice of kernel and bandwidth. Nevertheless, the first approach estimates in an ``efficient" way the transition at horizon 1 and we use this estimator to compute the IRF. The second approach is mixing an ``efficient" estimator of transition $g$, with an estimator $\hat{m}_T^{(h-1)}$, that does not account for its known dependence with respect to $g$. Therefore, we expect the second approach to be less accurate asymptotically under coherent choices of kernel and bandwidth in the two approaches. We show in the next subsection that this is not the case. 

\subsection{Asymptotic Properties}

Let us now derive the asymptotic behaviours of the functional estimators of interest that are: $\hat{g}_T(\cdot)$, $\widehat{IRF}_T(h,\cdot|y_t)$, $\widehat{\widehat{IRF}}_T(h,\cdot|y_t)$. All these estimators are consistent with respect to their theoretical counterparts and  they converge at different speeds. We perform in Appendix A the asymptotic expansions of the estimated IRF's in terms of the nonparametric estimates of the primitive characteristic of the conditional transitions. This allows for deriving the asymptotic behaviour of the estimated IRFs. For expository purposes, the derivations are performed in the one dimensional case and for stationary processes\footnote{The asymptotic results could be extended to nonstationary Markov processes if they satisfy a null recurrence property [Karlsen and Tjostheim (2011), Gourieroux and Jasiak (2019)]. This is for instance the case of the Gaussian random walk. Then the speed of convergence is modified, will depend on the number of regenerations in the sampling period, but not on the initial value $y_0$ of the process. We do not develop this extension in this paper [see Wright (2000), Gospodinov (2004), Pesavento and Rossi (2007), Montiel-Olea and Plagborg-M\o ller (2021) for the effect of unit roots in the linear local projection, which is the analogue in the semi-parametric linear framework].}. Then, the asymptotic variances are functions of the conditional cumulative distribution functions (c.d.f.), $F(\mathfrak{z}|y)$, the transition density, $f(\mathfrak{z}|y)$, and the conditional quantile, $Q(\alpha|y)$. \\

\begin{proposition} The nonparametric direct and local projection estimators have the following asymptotic properties: 	\\
	
	(i) Direct Estimation \\
	
	Let us assume $T \rightarrow \infty$, $S \rightarrow \infty$, with $S/T \rightarrow 0$, $b_T \rightarrow 0$, with $Tb_T^{5/3} \rightarrow \infty$. Then, $\widehat{IRF}_T(h,\delta|y_t)$ converges to $IRF(h,\delta|y_t)$ for any $h,\delta$. Morever, we have the convergence in distribution:
	\begin{equation*}
		\sqrt{Tb_T}(\widehat{IRF}_T(h,\delta|y_t)-IRF(h,\delta|y_t)) \rightarrow^d N(0,\sigma^2(h,\delta|y_t)),
	\end{equation*}
where the expression of $\sigma^2(h,\delta|y_t)$ is derived in Appendices A.1 for $h=1$ and A.2 for $h \geq 2$. \\

(ii) Local Projection \\

We have: 

\begin{equation*}
	\sqrt{Tb_T}[\widehat{\widehat{IRF}}_T(h,\delta|y_t)-\widehat{IRF}_T(h,\delta|y_t)] = o_p(1),
\end{equation*}
In particular, the NLP estimator $\widehat{\widehat{IRF}}_T(h,\delta|y_t)$ converges to $IRF(h,\delta|y_t)$, with the same speed of convergence as the direct estimator $\widehat{IRF}(h,\delta|y_t)$ and they have the same asymptotic distribution and asymptotic variance. \\ 
\end{proposition}

\textbf{Proof:} See Appendices A.1, A.3, A.4 for the expansions and online appendix D for regularity conditions. \\

Recall the standard arguments usually provided in semi-parametric linear dynamic models when comparing the direct and indirect approaches. If the semi-parametric model (i.e. the lag) is well specified, the direct approach is more efficient than the local projection approach. If the semi-parametric model (i.e. the lag) is misspecified (for instance, if the true model is a VAR$(\infty)$), the local projection approach (applied with a sufficiently large controlled lag order \footnote{See Xu (2023), Montiel-Olea et al. (2024), for statistical inference for local projections when the controlled lag order diverges.} \footnote{By increasing the number of lags in the nonlinear local projection approach, we will encounter the curse of dimensionality on nonparametric approaches (see discussion in Section 5).}) is better. The proposition above shows that, under the Markov assumption, both approaches are consistent and asymptotically equivalent in a pure nonparametric approach, for both linear and nonlinear dynamic models. Moreover, they have the same speed of convergence and the same asymptotic distribution. Therefore, neither nonlinear local projections, nor nonlinear autoregressions dominates\ the other in terms of asymptotic inference. They can differ however by their finite sample properties, and also by the computational time that they require. \\

The results of Proposition 3 can be compared with the standard results when the IRF is computed by linear local projection (LP). In this linear semi-parametric framework, the associated multi-step error forecasts are generally serially correlated [see Montiel-Olea and Plagborg-Moller (2021), p 1793] and the standard errors of the IRF are usually estimated by heteroscedasticity and autocorrelation (HAC) consistent approaches [Jorda (2005), Ramey (2016), Kilian and Lutkepohl (2017)]. Such an adjustment is not needed with the nonparametric approaches, since the functional estimators are not only local in $h$, but also local in the value of the conditioning variable. 
 
\subsection{An Illustration}

This section illustrates the above concepts through the lens of the DAR(1) model in Example 1. In particular, we consider the data generating process:
\begin{equation*}
	y_t = 0.5 y_{t-1} + \sqrt{1 +0.5 y_{t-1}^2} \ \varepsilon_t,
\end{equation*}
where $y_0=0.2$ and $\varepsilon_t \sim N(0,1)$. These parameter values ensure that there exists a strictly stationary solution with second-order moments. 

\subsubsection{The Data Generating Process}

We simulate 200 observations of the DAR(1) process described above and plot both its trajectory and empirical density in Figure 1. The empirical density features rather fat tails, that is a standard effect of conditional heteroscedasticity on the marginal (i.e. stationary) distribution.
\begin{figure}[h]
	\centering
	\includegraphics[width=1\linewidth]{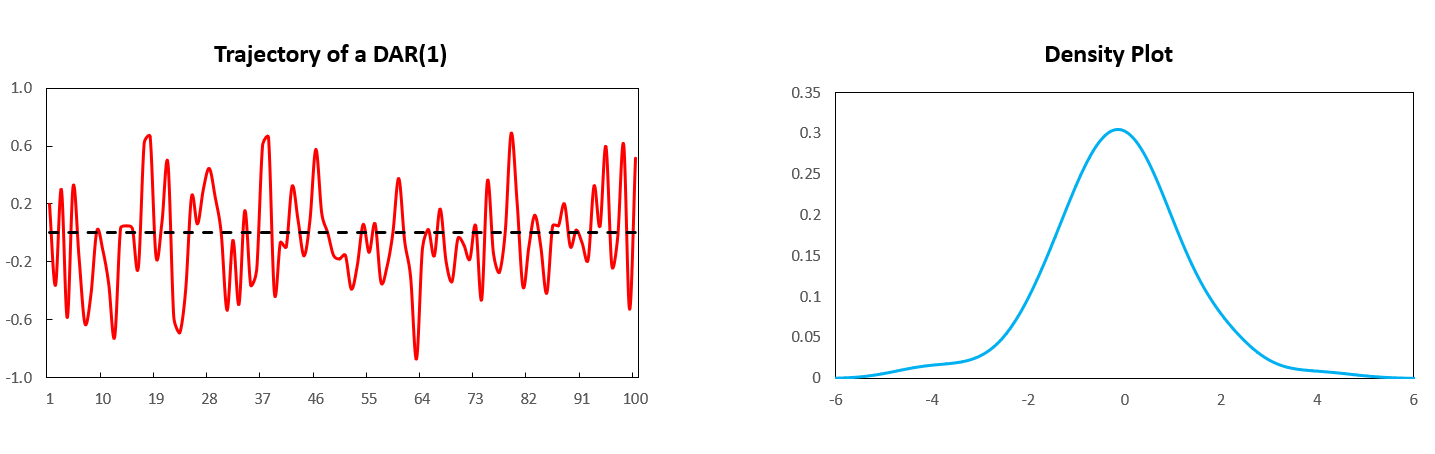}
	\caption{The trajectory (left figure) and empirical density (right figure) of the DAR(1) model.}
\end{figure}

\subsubsection{Estimators}

The DAR model is semi-parametric, since its parameters and error distribution are not known. For estimation of the IRF, we first need to estimate the function $g$. To do so, we first perform a Quasi Maximum Likelihood Estimation (QMLE) for the parameters $\rho$, $\alpha$ and $\beta$, that is a MLE as if the $\varepsilon_t$ were Gaussian. The solution to the quasi log-likelihood can be written as [Ling (2007)]: 
\begin{equation*}
(\hat{\rho},\hat{\alpha},\hat{\beta})' =	\argmax_{\rho,\alpha,\beta} \sum_{t=2}^{200}-\frac{1}{2}\left[\ln(\alpha + \beta y_{t-1}) - \frac{1}{2}\frac{(y_t - \rho y_{t-1})^2}{(\alpha + \beta y_{t-1})}\right].
\end{equation*}
We perform a three dimensional grid search on the cube $[0.01,1.20]^3$ at intervals of size 0.01. Our results yield estimates $(\hat{\rho}_{MLE},\hat{\alpha}_{MLE},\hat{\beta}_{MLE})'= (0.35,0.9,0.46)'$, to be compared with the true values of the DGP, that are $(0.5, 1, 0.5)$. 
\begin{figure}[h]
	\centering
	\includegraphics[width=1\linewidth]{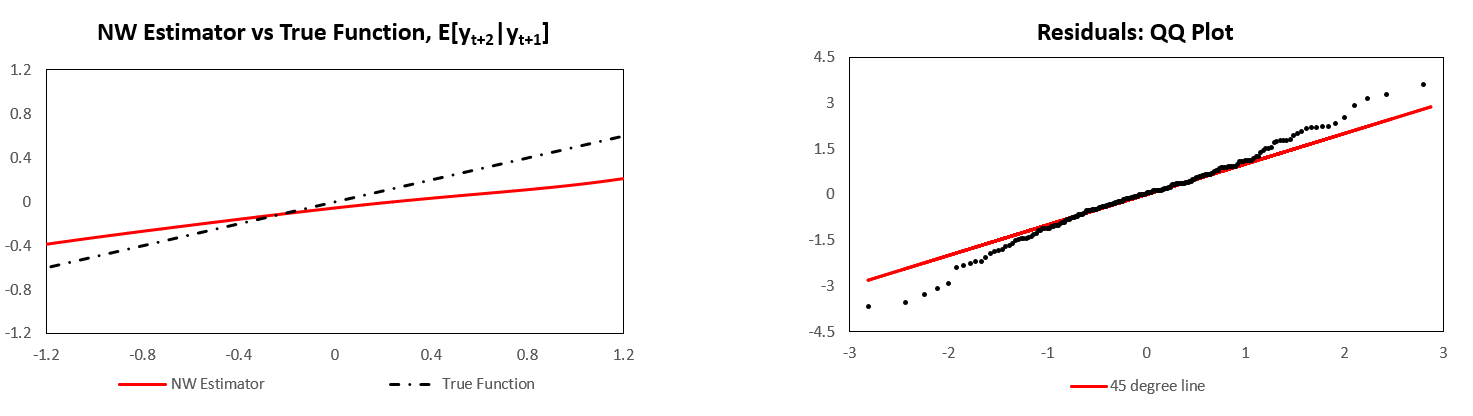}
		\caption{The comparison of the estimated and true functional form of $m$ (left figure) and the QQ plot of the residuals (right figure). }
\end{figure}
For the local projection approach, we will have to run the nonparametric regression of $y_{t+h}$ on $y_t$ to obtain the Nadaraya-Watson estimates $\hat{m}_T^{(h)}$ in the NLP equation. In Figure 2, we compare the estimated regression line $\hat{m}_T^{(1)}$ and the true function, as well as the QQ plot of the residuals. 

\subsubsection{IRF}

We present shocks of magnitude $-1, -0.5, 0.5$ and $1$ below in Figure 3. Each graph features three IRFs: [1] The true IRF based on equation \eqref{IRF_pointwise}; [2] The IRF obtained by the (semi-parametric) direct estimation based on equation \eqref{dd_irf} with QMLE estimation; [3] The IRF obtained by means of nonlinear local projection based on equation \eqref{lp_irf}. We see that, in all cases, holding the simulation size equal, the direct estimation method provides a more accurate IRF. This is not surprising since the nonlinear local projection is based on the Nadaraya-Watson estimator, which takes longer to converge than the QMLE method. 

\begin{figure}[h]
	\centering
	\includegraphics[width=1\linewidth]{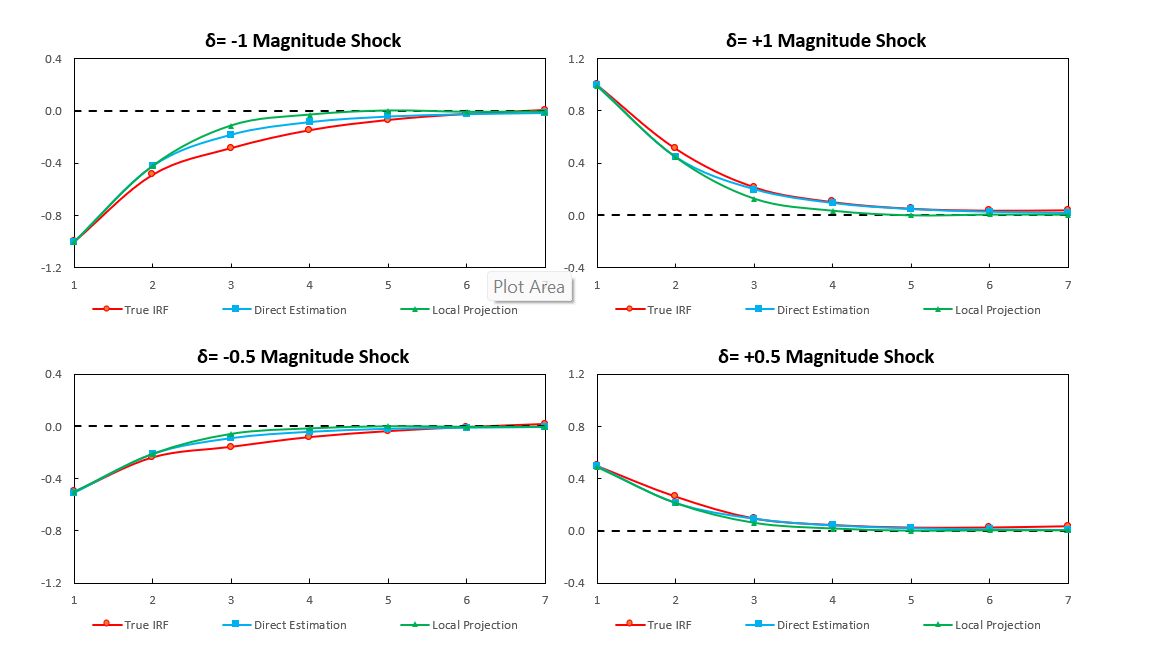}
	\caption{A comparison of IRFs constructed using different estimation methods.}
\end{figure}

\section{Statistical Inference for Multivariate Case}

The nonparametric approaches of Section 4 cannot be easily extended to the multivariate framework with a dimension $n$ larger or equal to 3. Indeed, we encounter the curse of dimensionality due to the conditioning variable in kernel estimation of conditional cdfs and conditional expectations. The curse of dimensionality will be reached even faster for Markov processes of order $p$, $p>1$. Indeed, the rate of convergence of these nonparametric estimators deteriorates quickly [i.e. at an exponential rate] as the dimension grows [see Geenens (2011), Conn and Li (2018)]. In this section, we consider semi-parametric models in which the nonparametric dimension is reduced\footnote{To circumvent the curse of dimenisonal for either $n$, or $p$ larger than 2, say, Goncalves et al. (2024 a,b) propose to modify the conditioning set. Instead of conditioning by $y_{t-1}$ (when $p=1$ for instance), they suggest to condition by a one dimensional function $a(y_{t-1})$ of $y_{t-1}$. This solves the curse of dimensionality issue, but at the cost of a loss in information, and a change of the definition of the IRF, which depends on the transformation $a(\cdot)$.  }. 

\subsection{Semiparametric Models}

Let us now provide additional examples of multivariate Markov processes. \\

\textbf{Example 3: Nonlinear Conditionally Gaussian Model}\\

It is written as: 
\begin{equation}
	y_t = m(y_{t-1};\theta) + D(y_{t-1};\theta)\varepsilon_t,	
\end{equation}
where $\varepsilon_t\sim IIN(0,Id)$ and is assumed independent of $y_{t-1}$. This specification is used in Koop, Pesaran and Potter (1996, eq.(1)) for instance\footnote{See also Goncalves et al. (2021), Ballarin (2025) without conditional heteroskedasticity}. Except in the Gaussian linear VAR(1) framework: $m(y_{t-1};\theta)=A(\theta) y_{t-1}, D(y_{t-1};\theta)=D(\theta)$, independent of the past, the process is not conditionally Gaussian at horizons larger or equal to 2. The nonlinear effects will appear on the IRF's after $h \geq 2$.  \\

For instance, the  multivariate extension of the DAR model can be written as:
\begin{equation*}
	y_t = \Phi y_{t-1} + (\text{diag}\ h_t)^{1/2}\varepsilon_t,
\end{equation*}
with $h_t=a + B(y^2_{1,t-1},...,y^2_{n,t-1})$ [see Zhu et al. (2017)]. \\

\textbf{Example 4:} \textbf{Strong Linear Causal Structural SVAR(1) Model} \\ 

The process is defined as the stationary solution to the linear difference equation:
\begin{equation}\label{slcvar}
	y_t = A y_{t-1} + D u_t,
\end{equation}
where the eigenvalues of autoregressive matrix $A$ have a modulus strictly smaller than 1, $D$ is an invertible matrix, and the $u_t$'s are i.i.d, with independent components. If the cdf of $u_{it}$, $i=1,...,n$, is $F_i$, assumed to be invertible, then this model can be rewritten with Gaussian errors as:
\begin{equation}\label{slcvarge}
	y_t = A y_{t-1} + D[F^{-1}_i \circ \Phi(\varepsilon_{i,t})], 
\end{equation}
where $\Phi$ is the cdf of the standard normal distribution. The strong linear SVAR(1) model \eqref{slcvarge} has in general a nonlinear dynamic feature when it is written with respect to Gaussian errors, by the nonlinear transformation $ F^{-1}_i \circ \Phi=Q_i \circ \Phi, \ i=1,...,n$, where $Q_i=F^{-1}_i$ is the quantile function. However, it is linear in $y_{t-1}$ and contains no cross-effects of $y_{t-1}$ and $\varepsilon_t$. We get a semi-parametric model that includes vector parameter $\beta=((\text{vec A})',(\text{vec D})')'$ and functional parameters $Q_i$ for $i=1,...,n$. These functional parameters depend on one-dimensional arguments, which allows us to circumvent the nonparametric curse of dimensionality. \\

Let us now discuss the respective roles of specifications \eqref{slcvar} and \eqref{slcvarge}. 
\begin{enumerate}
	\item If $u_t=\varepsilon_t$, it is well-known that the parameter $A$ is identifiable, whereas $D$ is not identifiable. Therefore, the IRF becomes identifiable only if $A(\theta)$, $D(\theta)$ are (jointly) parameterized with identifiable $\theta$. The IRF has the simple form:
	\begin{equation*}
		IRF(h,\delta)=D\delta + ... + A^hD\delta = (Id-A)^{-1}(Id-A^{h+1})D\delta.
	\end{equation*}
	\item If $u_t$ differs from $\varepsilon_t$ and if at most one component of $u_t$ is Gaussian, it is known that the model \eqref{slcvar} is semi-parametrically identifiable (up to permutation of indexes), that is parameters $A,D$ and functional parameters $F_i(\cdot)$, $i=1,...,n$ are identifiable by applying the identification results in linear ICA [see Comon (1994) for the linear ICA, and Gourieroux, Monfort and Renne (2017) for the application to SVAR models]. 
\end{enumerate}

The transformations of the error term $u_t$ into a Gaussian error $\varepsilon_t$, when the components of $u_t$ are independent, can also be used to extend the models of Example 3 to non-Gaussian errors and then render semi-parametric the models in Example 3. Therefore, the reduction of the curse of dimensionality is due to the assumption of independent components for $u_t$, that is the possibility to  replace the nonparametric estimator of the joint density of $u$ by the nonparametric estimation of the marginal distributions $F_i$, or $Q_i$, for $i=1,...,n$.

\subsection{Direct vs Indirect Estimation}

Let us now introduce the different (semi-) parametric estimation approaches of the model parameters, then of the IRF. 

\subsubsection{Nonlinear Structural VAR (NSVAR)}

Let us consider the nonlinear autoregressive model:
\begin{equation}\label{nsvar}
	y_t = g(y_{t-1},u_t;\beta),
\end{equation}
where the components of the errors are $u_{i,t}=F_i^{-1}\cdot \Phi (\varepsilon_{i,t})$, $i=1,...,n$, and $(\varepsilon_t)$ is $IIN(0,Id)$. Then, model \eqref{nsvar} can be written under the equivalent form:
\begin{equation}\label{eqnsvar}
	G(y_t,y_{t-1};\beta) = u_t,
\end{equation}
where $u_t = F_i^{-1}\circ \Phi (\varepsilon_{i,t})=Q_i\circ \Phi(\varepsilon_{i,t})$, $G$ is a demixing transformation, that is, the inverse of function $g$ with respect to $u_t$, and $\beta$ is a parameter. Then we can consider a parametric and a semi-parametric version of this model: \\ 

 \textbf{1. Parametric Model:} \\
	
	When the distributions $F_i$, $i=1,...,n$, are parameterized, with parameter $\alpha$, the model \eqref{nsvar} and \eqref{eqnsvar} becomes parametric with global parameters $\theta'=(\alpha',\beta')$. It can be estimated by maximum likelihood for instance. We denote $\hat{\theta}_T$ the corresponding estimator. \\
	
\textbf{2. Semi-Parametric Model:}\\ 

Under the assumption that the parameter $\beta$ is identifiable, it is possible to construct a consistent and asymptotically normal estimator of $\beta$. Then, this estimator satisfies: 
\begin{equation}\label{abeta}
	\sqrt{T}(\hat{\beta}_T-\beta_0) \rightarrow N(0,V(\beta_0)),
\end{equation}
for instance by applying a minimization of the cross covariances and autocovariances of nonlinear functions $G(y_t,y_{t-1},\beta)$ [see Gouri\'eroux and Jasiak (2023), Velasco (2023), for Generalized Covariance (GCov) estimators]. Once the parameter $\beta$ is estimated, we can compute the estimated errors as: $\hat{u}_{t,T} = G(y_t,y_{t-1};\hat{\beta}_T)$, $t=1,...,T$, and deduce nonparametric estimators $\hat{Q}_{i,T}$ of $Q_i$, for $i=1,...,n$ by applying the kernel approach \eqref{kapp} to the series of residuals. 

\subsubsection{Direct Estimation}

 \textbf{1. Parametric Model:} \\

The estimated IRF is obtained by plugging in the theoretical expression of the IRF the maximum likelihood estimator $\hat{\theta}_T$ of $\theta$: 
\begin{equation}
	\widehat{IRF}_T(h,\delta|y_t) = IRF(h,\delta,\hat{\theta}_T|y_t).
\end{equation}

\textbf{2. Semi-Parametric Model:}\\ 

For exposition, let us consider the case $h=1$. In the semi-parametric framework, the IRF becomes a function of parameter $\beta$ and functional parameters $Q_i$, $i=1,...,n$, $IRF(h,\delta;\beta, (Q_i)|y_t)$, that is:
\begin{equation}\label{sp_irf_true}
	\begin{split}
		&	IRF(1,\delta;\beta,(Q_{i})) \\
		= & \mathbb{E}\left[g(y_t,u_{t+1}(\delta);\beta)-g(y_t,u_{t+1};\beta)\right] \\ 
		= & \mathbb{E}\left[g(y_t,\text{vec}(Q_i\circ\Phi(\varepsilon_{i,t+1}+\delta));\beta)-g(y_t,\text{vec}(Q_i\circ\Phi(\varepsilon_{i,t+1}));\beta)\right] \\ 
		=& \int ... \int\left[ g(y_t,\text{vec}(Q_i\circ\Phi(\varepsilon_{i}+\delta));\beta)-g(y_t,\text{vec}(Q_i\circ\Phi(\varepsilon_{i}));\beta)\right]\phi(\varepsilon_1)...\phi(\varepsilon_n)d\varepsilon_1...d\varepsilon_n,\\
	\end{split}
\end{equation}
where $\text{vec}(a_i)$ denotes the $n$-dimensional vector with components $a_i$. When the model is well-specified with true parameters $\beta_0$, $Q_{i,0}$, $i=1,...,n$, the true IRF is $IRF(1,\delta;\beta_0,(Q_{i,0}))$. It can be estimated by plugging in the estimator $\hat{\beta}_T$ of $\beta$ and the functional estimator $\hat{Q}_{i,T}$ of $Q_i$ for $i=1,...n$. The resulting estimator $IRF(1,\delta,\hat{\beta}_T,(\hat{Q}_{i,T}))$ is consistent, asymptotically normal, and its asymptotic distribution is obtained by the delta method applied to both parameter $\beta$ and functional parameters $Q_i$, $i=1,...,n$ [see Appendix A.5]. 
\begin{proposition} 
	Under standard regularity conditions, the asymptotic distribution of the estimated IRF is such that: 
\begin{equation}\label{delta_irf}
	\sqrt{Tb_t}\left(IRF\left[1,\delta;\hat{\beta}_T,(\hat{Q}_{i,T})\right]-IRF\left[1,\delta;\beta_0,(Q_{i,0})\right]\right)\xrightarrow[]{d}N(0,V(1,\delta,\beta_0,(Q_{i,0}))),
\end{equation}
where $V(1,\delta,\theta_0,(Q_{i,0}))$ is the asymptotic variance given in Appendix A.5.
\end{proposition}

\textbf{Proof:} See Appendix A.5.\\

\textbf{Remark 4:} The limiting distributions in \eqref{delta_irf} have been written for any given pair $h$,$\delta$. They can be extended to several pairs by taking into account the asymptotic covariances between $\widehat{IRF}_T(h,\delta|y_t)$ and $\widehat{IRF}_T(h^*,\delta^*|y_t)$ for two pairs $(h,\delta)$ and $(h^*,\delta^*)$. This joint inference can be used to provide confidence bands for the term structure of the IRF for a given $\delta$, and/or for the confidence bands of the IRF function of the magnitude of the shock for a given horizon\footnote{See Inoue et al. (2024), Section 5.2, for ``simultaneous inference" in a linear dynamic setting.}. \\

In the nonlinear autoregressive model, the transition at horizon larger than 2 has no closed form expression and has to be approximated by simulation based on the estimated model $y_t = g(y_{t-1},\varepsilon_t,\hat{\beta}_T,(\hat{Q}_{i,T}))\equiv \hat{g}_T(y_{t-1},\varepsilon_t)$. This is done along the same steps as in Section 4.1. The number $S$ is chosen by the econometrician. If $S$ is chosen much larger than the number of observations $T$, that is if $\frac{S_T}{T}$ tends to infinity with $T$, then the asymptotic behaviour of the estimated IRF based on simulations is the same as in \eqref{delta_irf}. \\

\textbf{Remark 5:} Instead of the recursion:
\begin{equation*}
	\hat{y}_{t+h}^s = g(\hat{y}_{t+k-1}^s,(\hat{Q}_{i,T})\left[\Phi(\varepsilon_{i,t+k}^s)\right],\hat{\beta}_T), k = 1,2,...,
\end{equation*} corresponding to the recursion (4.3), it would be possible to apply a ``bootstrap" recursion of the form:
\begin{equation*}
	\hat{y}_{t+k}^s = g(\hat{y}_{t+k-1}^s,\hat{u}_{t+k,T}^s,\hat{\beta}_T), k = 1,2,...,
\end{equation*} where the $\hat{u}^s_{t+k,T}$ are independently drawn among the residuals of $\hat{u}_{t,T}$, $t=1,...,T$. The asymptotic properties of this bootstrap version however, is beyond scope of this paper and left for future research.

\subsubsection{Indirect Approach (Local Projection) for n =2}

Due to the curse of dimensionality of nonparametric approaches, the nonlinear local projection approach will faces challenges for a dimension strictly greater than 2 given the number of observations usually available in macroeconomic applications. If $n=2$, we can use a local projection approach based on the formula of Proposition 2 in which the short term responses $g(y_t,\varepsilon_t+\delta;\theta)$, $g(y_t,\varepsilon_t;\theta)$ are estimated parametrically, whereas the conditional expectation $m^{(h-1)}(\cdot)$ is estimated nonparametrically by the Nadaraya-Watson approach, say. This approach avoids the simulation based parametric approach used in Section 5.1.1 to approximate $m^{(h-1)}(\cdot)$, but with the drawback of a nonparametric rate of convergence infinitely smaller than the rate of convergence in the direct parametric and semi-parametric approach in Section 5.2.2. Under the assumption of a well-specified model, the rates of convergence for the estimated IRF to their true values will be the parametric rate of $1/\sqrt{T}$ in a parametric direct approach, the nonparametric rate of $1/\sqrt{Tb_T}$ for one-dimensional kernel estimators in the semi-parametric direct approach and the nonparametric rate of two dimensional kernel estimators in the indirect local projections approach\footnote{This curse of dimensionality problem is also encountered in the multivariate VAR model in which the local projection implies a regression with a rather large number of control variables and is sometimes approximated by sparsified LASSO [see Adamek et al. (2024)].}. To summarize, the nonlinear local projection approach cannot profit from the reduced nonparametric dimensionality of the semi-parametric model (5.4).

\section{Illustrations}

Let us now illustrate  the estimation approaches in a bivariate nonlinear dynamic model.

\subsection{The Simulated Data}

Let us consider a bivariate Double Autoregressive model defined by:
\begin{equation*}
	Y_t = \Phi Y_{t-1} + \begin{bmatrix}
		\sqrt{h_{1,t}}u_{1,t} \\
		\sqrt{h_{2,t}}u_{2,t} \\
	\end{bmatrix},
\end{equation*}
where $h_t = b+ A\begin{bmatrix}
	y^2_{1,t-1} \\
	y^2_{2,t-1} \\
\end{bmatrix}$ and the same experiment as in Zhu et al. (2017), Section 3. The parameter $\beta=\left[(\text{vec $\Phi$})',b',(\text{vec A})'\right]'$ is fixed at $\beta = (0.4,0.1,-0.3,0.4,0.1,0.2,0.3,0.2,0.1,0.4)'$, $u_{1,t}$ and $u_{2,t}$ are independent with distributions $N(0,1)$ and student $t(4)$, respectively. The number of observations is fixed to $T=400$, and we plot the simulated series below in Figure 4. 

\begin{figure}[h]
	\centering
	\includegraphics[width=1\linewidth]{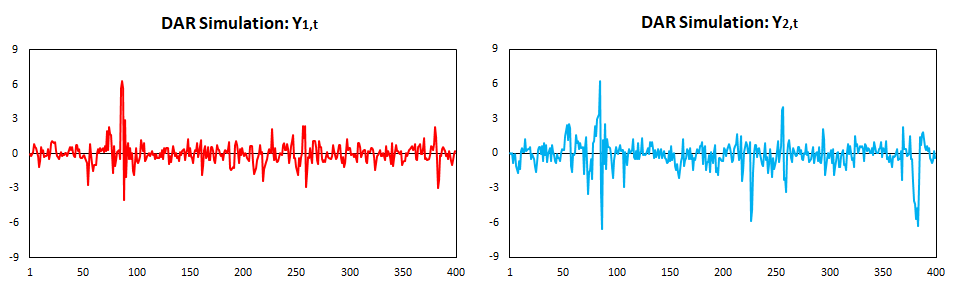}
	\caption{The simulated series}
\end{figure}

\subsection{Semi-Parametric Estimation}

The parameter $\beta$ is estimated in two steps. First, we regress by OLS $Y_t$ on $Y_{t-1}$, that provides an estimator $\hat{\Phi}_T$ of $\Phi$ and deduce $\hat{\upsilon}_{t,T} = y_t -\hat{\Phi}_Ty_{t-1}$. In a second step, we regress by OLS the squared $\upsilon$'s, i.e. $\begin{bmatrix}
	\hat{\upsilon}^2_{1,t,T} \\
	\hat{\upsilon}^2_{2,t,T} \\
\end{bmatrix}$ on  $\begin{bmatrix}
y^2_{1,t-1} \\
y^2_{2,t-1} \\
\end{bmatrix}$ with intercept. The results of these regressions are tabulated below in Tables 1 and 2, with the usual OLS standard errors. Note that these standard errors are not adjusted for conditional heteroscedasticity (in both regressions), for the effect of the first step estimation in the second regression, and the fact that the student distribution t(4) has no fourth-order moments. As a consequence, the standard errors in Table 2 are likely underestimated. \\
\begin{table}[h!]
	\centering
	\begin{minipage}[t]{0.45\textwidth}
		\centering
		\caption{OLS $Y_t$ on $Y_{t-1}$}
		\begin{tabular}{lcc}
			\hline
			\hline
			& $y_{1,t}$ & $y_{2,t}$ \\
			\hline
			$y_{1,t-1}$ & 0.30740 & -0.39328 \\
			& (0.04824) & (0.05925) \\
			$y_{2,t-1}$ & 0.01030 & 0.39006 \\
			& (0.03535) & (0.04341) \\
			\hline
			\hline
		\end{tabular}
	\end{minipage}%
	\hfill
	\begin{minipage}[t]{0.45\textwidth}
		\centering
		\caption{OLS $\hat{\upsilon}^2_{t,T}$ on $Y_{t-1}^2$}
		\begin{tabular}{lcc}
			\hline
			\hline
			& $\hat{\upsilon}^2_{1,t,T}$ & $\hat{\upsilon}^2_{2,t,T}$ \\
			\hline
			$y^2_{1,t-1}$ & 0.52377 & 0.13618 \\
			& (0.03567) & (0.02384) \\
			$y^2_{2,t-1}$ & 0.13180 & 0.43814 \\
			& (0.05274) & (0.03525) \\
			\hline
			\hline
		\end{tabular}
	\end{minipage}
\end{table}

Using the results for $\hat{b}_T$, $\hat{A}_T$, and $\hat{v}_{t,T}$, we compute: 
\begin{equation*}
\hat{u}_{t,T}=\left[\textup{diag}\left(\hat{b}_T+\hat{A}_T\begin{pmatrix}
	y^2_{1,t-1} \\
	y^2_{2,t-1} \\
\end{pmatrix}\right)\right]^{-1/2}(y_t-\hat{\Phi}_Ty_{t-1}),
\end{equation*}
 for $t=2,...,T.$ These are used to derive the estimated quantile functions $\hat{Q}_{i,T}$, $i=1,2$. We provide below in Figure 5 their Q-Q plots with respect to the standard normal distribution. 
 
 \begin{figure}[h]
 	\centering
 	\includegraphics[width=1\linewidth]{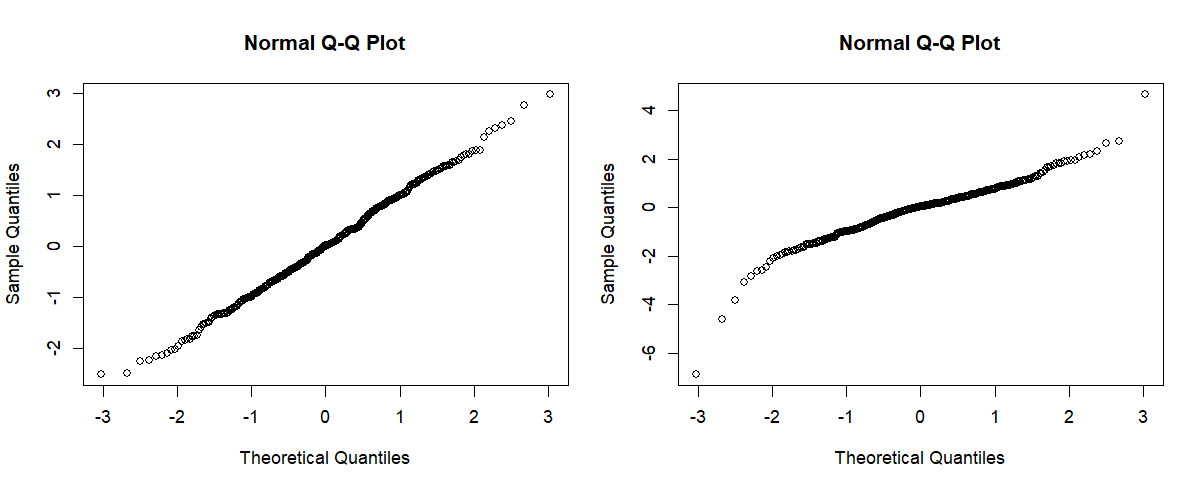}
 	\caption{QQ plot with respect to the normal distribution. On the left is the QQ plot for $\hat{u}_{1,t,T}$ and on the right is the QQ plot for $\hat{u}_{2,t,T}$.}
 \end{figure}

As expected, the values in the QQ plot for $\hat{u}_{1,t,T}$ lie close to the 45 degree line since the residuals are normally distributed. On the other hand, the QQ plot for $\hat{u}_{2,t,T}$ deviate from the 45 degree line further from the mean due to the fatter tails of the t-distribution.

\subsection{The ``True" and Estimated IRF}

To visualize the nonlinear local projection in the context of the bivariate DAR model, we consider four different shock scenarios with the same magnitude but different signs: $\delta = (+0.5,+0.5)$, $\delta = (+0.5,-0.5)$, $\delta = (-0.5,+0.5)$ and $\delta = (-0.5,-0.5)$. The results are featured in Figure 6 below. The red lines and blue lines represent the responses of the first and second variables, respectively. The dotted and solid lines represent the ``true" IRFs and the estimated IRFs, respectively. Note that a larger sample of 3000 observations were used to generate nonlinear local projections. This reveals that the nonparametric approach is computationally expensive and requires a rather large data set on hand. For instance, even in the case of a $\delta=(+0.5,-0.5)$ shock to the process, the estimated IRF is still far away from the truth.\\

The behaviours of the plots showcase the nonlinear dynamics at play for the bivariate DAR process. We again see the absence of shock symmetry, as a positive and negative shock of the same magnitude to not seem to mirror one another in all scenarios. We also see that the response of $Y_{1,t}$ and $Y_{2,t}$ are different for shocks of the same magnitude and direction. This reflects not only the different parameterization for each variable in the DGP, but also the presence of non-Gaussianity in the second component of the innovation. 
\begin{figure}[h!]
	\centering
	\includegraphics[width=0.8\linewidth]{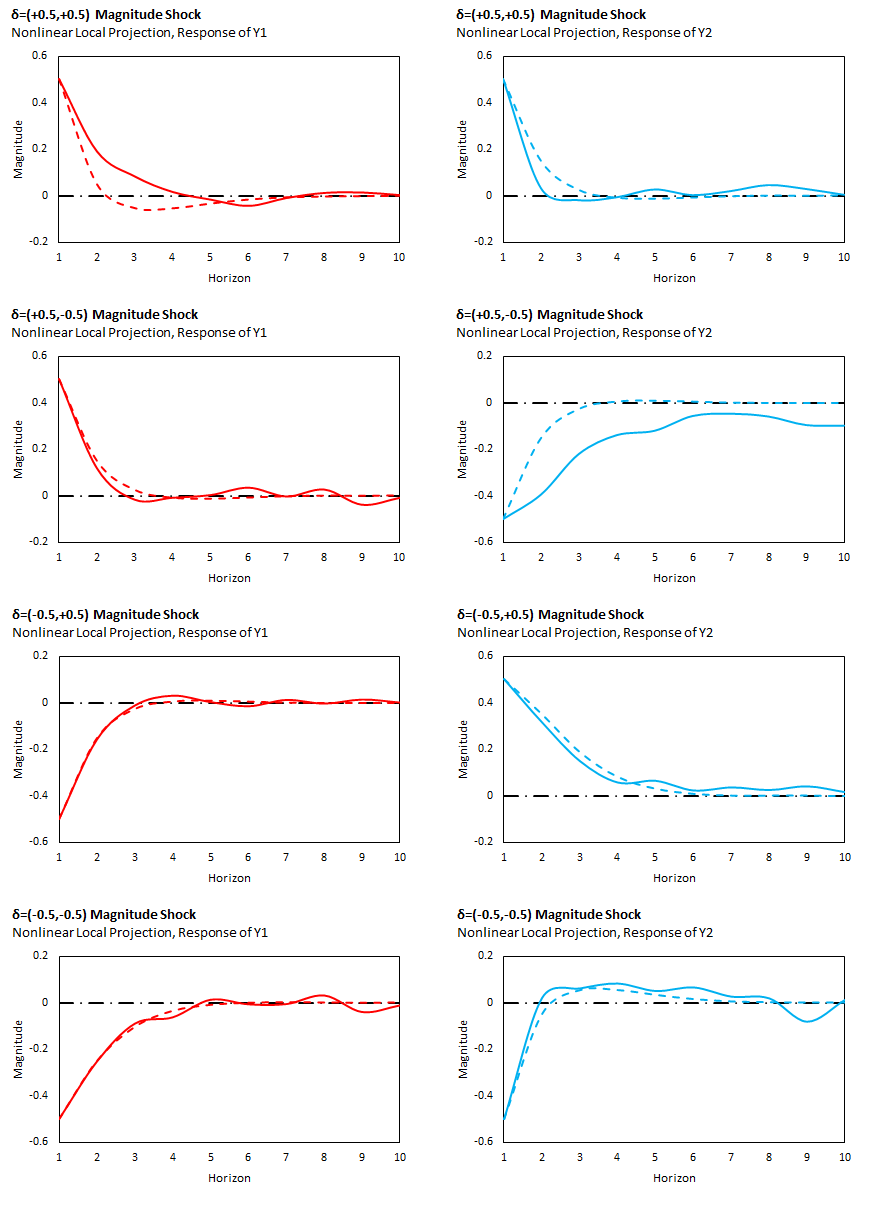}
	\caption{Four shock scenarios with $\delta = (+0.5,+0.5)$, $\delta = (+0.5,-0.5)$, $\delta = (-0.5,+0.5)$ and $\delta = (-0.5,-0.5)$. The red lines represent the response of the first variable $Y_{1,t+h}$, and the blue lines represent the response of the second variable $Y_{2,t+h}$. The dotted lines denote the ``true" IRFs implied by the DAR model, and the solid lines are the IRFs implied by the nonlinear local projection. }
	\label{fig:irf1}
\end{figure}

\newpage

\section{Concluding Remarks}

This paper has extended the notions of innovations, impulse response functions and local projections to the nonlinear dynamic framework. In particular, we have precisely defined the nonlinear Gaussian innovation and the IRFs corresponding to shocks on these innovations and derive an integral expression of the IRF that extends to a nonlinear dynamic framework the idea of local projection. Then this integral expression can be used to extend the idea of local projection at least for small dimension of the variable to be predicted. In the one-dimensional framework, we have precisely compared the asymptotic distributions of the IRF estimated nonparametrically by direct and an indirect local  projection approach, respectively, and show that they are asymptotically equivalent. This generalizes a result already derived in the literature for linear dynamic models, based on the Frisch-Waugh-Lovell Theorem. For a larger dimension, we can encounter the curse of dimensionality of conditional nonparametric inference, especially to apply the nonlinear analogue of local projections. Then we extend the analysis to a special class of semi-parametric nonlinear dynamic models to circumvent the curse of dimensionality. This semi-parametric family includes, in particular, models with conditional means, conditional heteroscedasticity and non-Gaussian errors. Nevertheless, only the case of bivariate models seem feasible in practice for nonparametric nonlinear local projection and the extensions of local projections for dimension 2 and semi-parametric analyses are infinitely less precise compared to direct parametric and even semi-parametric approaches in nonlinear structural autoregressive models. \\

\newpage
	
\section{References}



Adamek, R., Smeekes, S., and I., Wilms (2024). Local Projection Inference in High Dimension. \textit{Econometrics Journal}. 27, 323-342. \\

Ait-Sahalia, Y. (1993). Nonparametric Functional Estimation with Applications to Financial Models. \textit{Doctoral Dissertation}, Massachusetts Institute of Technology. \\


Ballarin, G. (2025): Impulse Response Analysis of Structural Nonlinear Time Series. \textit{DP University of St. Gallen}. \\

Baqace, D. (2020). Asymmetric Inflation Expectations, Downward Rigidity of Wages, and Asymmetric Business Cycles. Journal of Monetary Economics, 114, 174-193.\\ 

Barnichon, R., and C., Brownlees. (2019). Impulse Response Estimation by Smooth Local Projections. \textit{Review of Economics and Statistics}, 101, 522-530. \\ 


Bec, F., Nielsen, H., and S., Saidi. (2020). Mixed Causal–Noncausal Autoregressions: Bimodality Issues in Estimation and Unit Root Testing. \textit{Oxford Bulletin of Economics and Statistics}, 82      , 1413-1428. \\ 



Billingsley, P. (1986). Probability and Measure, Wiley, New York. \\

Blanchard, O., and D., Quah. (1989). The Dynamic Effects of Aggregate Demand and Supply Disturbances. \textit{American Economic Review}, 73, 655-673. \\ 

Borkovec, M., and C., Kluppelberg. (2001). The Tail of the Stationary Distribution of an Autoregressive Process with ARCH(1) Errors. \textit{Annals of Applied Probability},  11, 1220-1241. \\ 

Cashin, P., Mohaddes, K., and  M., Raissi. (2017). Fair Weather or Foul? The Macroeconomic Effects of El Niño. \textit{Journal of International Economics}, 106, 37-54. \\ 

Chang, P. and S., Sakata. (2007). Estimation of Impulse Response Functions Using Long Autoregressions. \textit{The Econometrics Journal}, 10, 453-463. \\ 

Christiano, L. (2012). Christopher A. Sims and Vector Autoregressions. \textit{The Scandinavian Journal of Economics}, 114, 1082-1104. \\ 

Christoffersen, P., Du, D., and  R., Elkamhi. (2017). Rare Disasters, Credit, and Option Market Puzzles. \textit{Management Science}, 63, 1341-1364.\\

Conn, D., and E., Li (2019). ``An Oracle Property of the Nadaraya-Watson Kernel Estimator for High Dimensional Nonparametric Regression," Scandinavian Journal of Statistics, 46, 735-764. \\

Comon, P. (1994). Independent Component Analysis, a New Concept?. \textit{Signal Processing}, 36, 287-314. \\ 

Cox, J., Ingersoll, J., and S., Ross. (2005). A Theory of the Term Structure of Interest Rates. \textit{Econometrica}, 53, 385-407. \\ 

De Truchis, G., Fries, S., and A., Thomas (2024). Forecasting Extreme Trajectories Using Semi-Norm Representations. \textit{DP Paris Dauphine University}. \\

Diercks, A., Hsu, A., and A., Tamoni. (2024). When it Rains it Pours: Cascading Uncertainty Shocks. \textit{Journal of Political Economy}, 132(2), 694-720.\\

Dufour, J.M., and E., Renault. (1998). Short Run and Long Run Causality in Time Series: Theory. \textit{Econometrica}, 66, 1099-1125. \\ 

Falk, M. (1985). Asymptotic Normality of the Kernel Quantile Estimator. \textit{The Annals of Statistics}, 13, 428-433.\\





Frank, M., and  T., Stengos. (1988). Chaotic Dynamics in Economic Time‐Series. \textit{Journal of Economic Surveys}, 2, 103-133.\\


Frisch, R. (1933). Propagation Problems and Impulse Responses in Dynamic Economies. \textit{Economic Essays in Honour of Gustav Cassel.} Allen and Unwin Ltd.  \\  

Frost, J., and R., van Stralen. (2018). Macroprudential Policy and Income Inequality.\textit{ Journal of International Money and Finance}, 85, 278-290.\\

Geenens, G. (2011). Curse of Dimensionality and Related Issues in Nonparametric Functional Regression. \textit{Statistics Survey}, 5, 30-43. \\

Giacomini, R. (2013). The Relationship Between DSGE and VAR Models. \textit{VAR models in Macroeconomics–New Developments and Applications: Essays in Honor of Christopher A. Sims}, 1-25. \\

Gallant, A., Rossi, P., and  G., Tauchen. (1993). Nonlinear Dynamic Structures. \textit{Econometrica}, 61, 871-907.\\


Goncalves, S., Herrera, A., Kilian, L., and E., Pesavento. (2021). Impulse Response Analysis for Structural Dynamic Models with Nonlinear Regressors. \textit{Journal of Econometrics}, 225, 107-130. \\ 

Gonçalves, S., Herrera, A., Kilian, L., and E., Pesavento. (2024a). State-Dependent Local Projections. \textit{Journal of Econometrics}, 244(2), 105702.\\

Goncalves, S., Herrera, A., Kilian, L., and E., Pesavento. (2024b). Nonparametric Local Projections. Working Paper, Federal Reserve Bank of Dallas. \\


Gospodinov, N. (2004). Asymptotic Confidence Intervals for Impulse Response Functions of Near-Integrated Processes. \textit{Econometrics Journal}, 7, 505-527.\\

Gouri\'eroux, C.,  and J., Jasiak. (2005). Nonlinear Innovations and Impulse Responses with Application to VaR Sensitivity. \textit{Annales d'Economie et de Statistique}, 78, 1-31. \\ 

Gouri\'eroux, C., and J., Jasiak. (2006). Multivariate Jacobi Process with Application to Smooth Transitions. \textit{Journal of Econometrics}, 131(1-2), 475-505. \\

Gourieroux, C., and J., Jasiak. (2010). Local Likelihood Density Estimation and Value‐at‐Risk.\textit{ Journal of Probability and Statistics}, 2010(1), 754-851.\\

Gouri\'eroux, C., and J., Jasiak. (2017). Noncausal Vector Autoregressive Process: Representation, Identification and Semi-Parametric Estimation, \textit{Journal of Econometrics}, 200, 118-134. \\

Gouri\'eroux, C., and J., Jasiak. (2019). Robust Analysis of the Martingale Hypothesis. \textit{Econometrics and Statistics}, 9, 17-41. \\

Gouri\'eroux, C.,  and J., Jasiak. (2022). Nonlinear Forecasts and Impulse Responses for Causal-Noncausal (S)VAR Models. \textit{arXiv preprint arXiv:2205.09922}. \\ 

Gouri\'eroux, C.,  and J., Jasiak. (2023). Generalized Covariance Estimator. \textit{Journal of Business and Economic Statistics}, 41, 1315-1327. \\ 

Gouri\'eroux, C., Jasiak, J., and A., Monfort (2020). Stationary Bubble Equilibria in Rational Expectation Models. \textit{Journal of Econometrics}, 218, 714-735. \\

Gouri\'eroux, C., Jasiak, J., and R., Sufana. (2009). The Wishart Autoregressive Process of Multivariate Stochastic Volatility. Journal of Econometrics, 150(2), 167-181.\\



Gourieroux, C., and Q., Lee. (2024). Forecast Relative Error Decomposition. arXiv:2406.17708.\\

Gourieroux, C., and Q., Lee. (2025). Identification of Impulse Response Functions for Nonlinear Dynamic Models. arXiv:2506.13531.\\

Gouri\'eroux, C., Monfort, A., Mouabbi, S., and J.P., Renne. (2021). Disastrous Defaults. \textit{Review of Finance}, 25, 1727-1772.\\

Gouri\'eroux, C., Monfort, A.,  and J.P., Renne. (2017). Statistical Inference for Independent Component Analysis: Application to Structural VAR Models. \textit{Journal of Econometrics}, 196, 111-126. \\ 




Hamilton, J. (2003). What is an Oil Shock?. \textit{Journal of Econometrics}, 113, 363-398. \\

Hamilton, J. (2011). Nonlinearities and the Macroeconomic Effects of Oil Prices. \textit{Macroeconomic Dynamics}, 15, 364-378. \\ 

Herbst, E. , and B., Johannsen. (2024). Bias in Local Projections. \textit{Journal of Econometrics}, 240(1), 105655. \\

Heston, S. (1993). A Closed Form Solution for Options with Stochastic Volatility with Applications to Bond and Currency Options. \textit{Review of Financial Studies}, 6, 327-343. \\

Hull, S., and A., White (1987). The Pricing of Options on Assets with Stochastic Volatilities. \textit{Journal of Finance}, 42, 281-300.\\

Hyvarinen, A., and P., Pajonen (1999). Nonlinear Independent Component Analysis: Existence and Uniqueness Results. \textit{Neural Networks}, 13, 411-430. \\

Hyvarinen, A., Sasaki, A., and R., Turner (2019): Nonlinear ICA Using Auxilary Variables and Generalized Contrastive Learning. \textit{Proceedings of the 22nd International Conference on Artificial Intelligence and Statistics}, 859-868. \\ 

Inoue, A., Jorda, O., and G., Kuersteiner. (2025). Inference for Local Projections. \textit{The Econometrics Journal}, February.\\


Jordà, Ò. (2005). Estimation and Inference of Impulse Responses by Local Projections. \textit{American Economic Review}, 95, 161-182.\\


Jordà, Ò., Schularick, M.,  and  A., Taylor. (2015). Leveraged Bubbles. \textit{Journal of Monetary Economics}, 76, 1-20. \\ 

Karlin, S., and H., Taylor. (1981). A Second Course in Stochastic Processes. \textit{Elsevier}. \\ 

Karlsen, M., and D., Tjostheim. (2001). Nonparametric Estimation on Null Recurrent Time Series. \textit{Annals of Statistics}, 29, 372-416. \\ 

Kilian, L.,  and Y., Kim. (2011). How Reliable are Local Projection Estimators of Impulse Responses?. \textit{Review of Economics and Statistics}, 93, 1460-1466. \\ 

Kilian, L.,  and H., Lutkepohl. (2017). Structural Vector Autoregressive Analysis. \textit{Cambridge University Press}. \\ 


Kilian, L.,  and R., Vigfusson. (2011). Nonlinearities in the Oil Price–Output Relationship. \textit{Macroeconomic Dynamics}, 15, 337-363.\\

Kilian, L., and R., Vigfusson. (2017). The Role of Oil Price Shocks in Causing US Recessions. \textit{Journal of Money, Credit and Banking}, 49, 1747-1776.\\

Kolesar, M., and M., Plagb\o rg-Moller (2024). Dynamic Causal Effects in a Nonlinear World: The Good, the Bad and the Ugly. DP Princeton University. \\

Koop, G. (1996). Parameter Uncertainty and Impulse Response Analysis. \textit{Journal of Econometrics}, 72, 135-149.\\ 

Koop, G., Pesaran, H.,  and S., Potter. (1996). Impulse Response Analysis in Nonlinear Multivariate Models. \textit{Journal of Econometrics}, 74, 119-147. \\ 

Kuiper, W.,  and A., Lansink. (2013). Asymmetric Price Transmission in Food Supply Chains: Impulse Response Analysis by Local Projections Applied to US Broiler and Pork Prices.\textit{ Agribusiness}, 29, 325-343.\\


Lanne, M., and P., Saikkonen. (2011). Noncausal Autoregressions for Economic Time Series. \textit{Journal of Time Series Econometrics}, 3, 1-39. \\ 



Lee, Q. (2025). Nonlinear Forecast Error Variance Decompositions with Hermite Polynomials. arXiv:2503.11416. \\

Ling, S. (2007). A Double AR(p) Model: Structure and Estimation. \textit{Statistica Sinica}, 17, 161-175. \\ 

Metcalf, G.,  and J., Stock. (2020). Measuring the Macroeconomic Impact of Carbon Taxes. \textit{AEA papers and Proceedings}, 100, 101-06. \\ 


Montiel Olea, J., and M., Plagborg‐Møller. (2021). Local Projection Inference is Simpler and More Robust Than You Think. \textit{Econometrica}, 89, 1789-1823. \\ 

Montiel Olea, J., and M., Plagborg‐Møller. (2022). Corrigendum: Local Projection Inference is Simpler and More Robust Than You Think. Online Manuscript. \\ 


Montiel-Olea, J., Plagborg-Moller, M., Qian, E., and C., Wolf (2024). Double Robustness of Local Projections and Some Unpleasant VARithmetic. DP Princeton University. \\

Paul, P. (2020). A Macroeconomic Model with Occasional Financial Crises.\textit{ Journal of Economic Dynamics and Control}, 112, 103830.\\ 

Pesaran, H.,  and Y., Shin. (1998). Generalized Impulse Response Analysis in Linear Multivariate Models.\textit{ Economics Letters}, 58, 17-29. \\ 

Pesavento, E., and B., Rossi. (2007). Impulse Response Confidence Intervals for Persistent Data: What Have We Learned? \textit{Journal of Economic Dynamics and Control}, 31, 2398-2412. \\ 


Plagborg‐Møller, M., and C., Wolf. (2021). Local Projections and VARs Estimate the Same Impulse Responses. \textit{Econometrica}, 89, 955-980. \\


Ramey, V. (2016). Macroeconomic Shocks and their Propagation. \textit{Handbook of Macroeconomics}, 2, 71-162. \\ 




Rosenblatt, M. (1952). Remarks on a Multivariate Transformation. \textit{The Annals of Mathematical Statistics}, 23, 470-472.\\



Sims, C. (1980). Macroeconomics and Reality. \textit{Econometrica}, 48, 1-48. \\ 


Toda, A. (2020). Susceptible-Infected-Recovered (SIR) Dynamics of Covid-19 and Economic Impact. \textit{arXiv Preprint.} arXiv:2003.11221. \\ 


Tweedie, R. (1975). Sufficient Conditions for Ergodicity and Recurrence of Markov Chains on a General State Space. \textit{Stochastic Processes and their Applications}, 3, 385-403. \\ 


Velasco, C. (2023). Identification and Estimation of Structural VARMA Models Using Higher Order Dynamics. \textit{Journal of Business \& Economic Statistics}, 41, 819-832.\\

Wang, X. (2019). A Long Run Risks Model with Rare Disaster: An Empirical Test in the American Consumption Data. \textit{Advances in Economics, Business and Management Research}, 68, 237-243.\\ 

Weiss, A. (1984). ARMA Models with ARCH Errors. \textit{Journal of Time Series Analysis}, 5, 129-143. \\ 


Wright, J. (2000). Confidence Intervals for Univariate Impulse Responses with Near Unit Root. \textit{Journal of Business and Economic Statistics}, 18, 368-373. \\ 

Xu, K. (2023). Local Projection Based Inference Under General Conditions," DP Indiana University. \\ 


Zhu, H., Zhang, X., Liang, X., and Y., Li. (2017). On a Vector Double Autoregressive Model. \textit{Statistics \& Probability Letters}, 129, 86-95.\\ 




\newpage
	
\appendix

\begin{appendices}

\section{Asymptotic Results}

Since the aim of this paper is more on the definitions and discussions of IRF, we just provide below the main steps in the expansions of the nonparametric and semi-parametric estimators of the IRFs. Some regularity conditions needed for such expansions are given in online Appendix D. 

\subsection{Direct Estimation}  

For expository purposes, we first consider the case where $h=1$. The estimated IRF is given by: 
\begin{equation*}
	\begin{split}
		\widehat{IRF}_T(1,\delta) & = \mathbb{E}[\hat{y}^{(\delta)}_{t+1}-\hat{y}_{t+1}|y_t]\\
		& = \mathbb{E}[\hat{g}_T(y_t,\varepsilon_{t+1}+\delta)-\hat{g}_T(y_t,\varepsilon_{t+1})|y_t]\\
		& = \int_{-\infty}^{\infty}[\hat{g}_T(y_t,\varepsilon+\delta)-\hat{g}_T(y_t,\varepsilon)]\phi(\varepsilon)d\varepsilon.\\  
	\end{split}
\end{equation*}
Hence, we have:
\begin{equation*}
	\widehat{IRF}_T(1,\delta) - IRF_T(1,\delta) = \int_{-\infty}^{\infty}[\hat{g}_T(y_t,\varepsilon+\delta)-g(y_t,\varepsilon+\delta)]-[\hat{g}_T(y_t,\varepsilon)-g(y_t,\varepsilon)]\phi(\varepsilon)d\varepsilon,
\end{equation*}
where $y_t = Q(\Phi(\varepsilon)|y_{t-1})$\footnote{Note that this formula is only valid for the one dimensional framework}. As discussed in Section 4.2, the asymptotic variance of this expression will depend on the conditional cdf, $F(\mathfrak{z}|y)$, the transition density, $f(\mathfrak{z}|y)$, and the conditional quantile, $Q(\alpha|y)$. In Online Appendix C, we discuss their respective nonparametric estimators and convergence.\\

To procced, it is useful to expand the term $\hat{g}_T(y_t,\varepsilon)-g(y_t,\varepsilon)$, given by the following Lemma: 
\begin{lemma}
	Let us assume that $\hat{F}_T(\mathfrak{z}|y)$ tends to $F(\mathfrak{z}|y)$ such that $h(T)[\hat{F}_T(\mathfrak{z}|y)-F(\mathfrak{z}|y)] \rightarrow X(\mathfrak{z}|y)$, where $1/h(T)$ is the speed of convergence and $X(\mathfrak{z}|y)$ is the limiting process, indexed by the conditioning value. Then: 
	\begin{equation*}
		\begin{split}
			h(T)\left[\hat{g}_T(y_t,\varepsilon)-g(y_t,\varepsilon)\right] & = h(T)(\hat{Q}_T-Q)(\Phi(\varepsilon)|y_t) \\ 
			& = - \frac{h(T)}{f[Q(\Phi(\varepsilon_t)|y_t)|y_t]} (\hat{F}_T-F)[Q(\Phi(\varepsilon_t)|y_t)|y_t] + o_p(1).\\
		\end{split}
	\end{equation*}
\end{lemma}

\textbf{Proof:} See the Online Appendix B for details. \\ 

By Lemma 1, we have that:
\begin{equation*}
	\begin{split}
		& h(T)[\widehat{IRF}_T(1,\delta) - IRF(1,\delta)]\\ 
		=&  \int_{-\infty}^{\infty}\left[\frac{-h(T)}{f[Q(\Phi(\varepsilon_t+\delta)|y_t)|y_t]} (\hat{F}_T-F)[Q(\Phi(\varepsilon+\delta)|y_{t})|y_{t}]\right.\\
		& +\left.\frac{h(T)}{f[Q(\Phi(\varepsilon_t)|y_t)|y_t]} (\hat{F}_T-F)[Q(\Phi(\varepsilon)|y_{t})|y_{t}]  \right]\phi(\varepsilon)d\varepsilon \\ 
	\end{split}
\end{equation*} 
The variance of this expression is given by:

\begin{equation*}
	\begin{split}
		&\mathbb{V}\left\{h(T)[\widehat{IRF}_T(1,\delta) - IRF(1,\delta)]\right\} \\
		=&\mathbb{V}\left\{\int_{-\infty}^{\infty}\left[\frac{-h(T)}{f(Q(\Phi(\varepsilon+\delta)|y_{t}))} (\hat{F}_T-F)[Q(\Phi(\varepsilon+\delta)|y_{t})|y_{t}]\right.\right. \\ 
		& \left.\left. +\frac{h(T)}{f[Q(\Phi(\varepsilon_t)|y_t)|y_t]} (\hat{F}_T-F)[Q(\Phi(\varepsilon)|y_{t})|y_{t}] \right]\phi(\varepsilon)d\varepsilon\right\} \\
	\end{split}
\end{equation*}
\begin{equation*}
	\begin{split}
		= & \int\int \text{Cov}\left\{ \frac{-1}{f[Q(\Phi(\varepsilon+\delta)|y_t)|y_t]} h(T)(\hat{F}_T-F)[Q(\Phi(\varepsilon+\delta)|y_{t})|y_{t}] \right.\\
		& + \frac{1}{f[Q(\Phi(\varepsilon)|y_t)|y_t]} h(T)(\hat{F}_T-F)[Q(\Phi(\varepsilon)|y_{t})|y_{t}], \\
		& + \frac{1}{f[Q(\Phi(\varepsilon'+\delta)|y_t)|y_t]} h(T)(\hat{F}_T-F)[Q(\Phi(\varepsilon'+\delta)|y_{t})|y_{t}] \\
		& + \left. \frac{-1}{f[Q(\Phi(\varepsilon')|y_t)|y_t]} h(T)(\hat{F}_T-F)[Q(\Phi(\varepsilon')|y_{t})|y_{t}]\right\}\phi(\varepsilon)\phi(\varepsilon')d\varepsilon d\varepsilon'\\
	\end{split}
\end{equation*}
By Donsker's Theorem, we know that:
\begin{equation*}
	h(T)\left[\hat{F}_T(\mathfrak{z}|y)-F(\mathfrak{z}|y)\right]\rightarrow BB(\mathfrak{z}|y),
\end{equation*}
where $BB(\mathfrak{z}|y)$ denotes the Brownian Bridge in $\mathfrak{z}$ indexed by the conditioning value, $y$. The standard Brownian Bridge has the covariance operator $w(z,z')=\min[F(z),F(z')]-F(z)F(z')$. Hence, by substituting $z=Q[\Phi(\varepsilon|y_t)|y_t]$, we get: 
\begin{equation*}
	\begin{split}
		&\mathbb{V}\left\{h(T)[\widehat{IRF}_T(1,\delta) - IRF(1,\delta)]\right\} \\
		= & \int\int\left[ -\frac{1}{f(Q(\Phi(\varepsilon+\delta|y_{t}))} \frac{1}{f(Q(\Phi(\varepsilon')|y_{t}))} w\left(Q(\Phi(\varepsilon+\delta)|y_{t}),Q(\Phi(\varepsilon')|y_{t})\right)\right. \\
		& +\frac{1}{f(Q(\Phi(\varepsilon+\delta)|y_{t}))} \frac{1}{f(Q(\Phi(\varepsilon'+\delta)|y_{t}))} w\left(Q(\Phi(\varepsilon+\delta)|y_{t}),Q(\Phi(\varepsilon'+\delta)|y_{t})\right) \\ 
		& +\frac{1}{f(Q(\Phi(\varepsilon)|y_{t}))} \frac{1}{f(Q(\Phi(\varepsilon')|y_{t}))} w\left(Q(\Phi(\varepsilon)|y_{t}),Q(\Phi(\varepsilon')|y_{t})\right) \\ 
		& \left.-\frac{1}{f(Q(\Phi(\varepsilon)|y_{t}))} \frac{1}{f(Q(\Phi(\varepsilon'+\delta)|y_{t}))} w\left(Q(\Phi(\varepsilon)|y_{t}),Q(\Phi(\varepsilon'+\delta)|y_{t})\right)\right]\phi(\varepsilon)\phi(\varepsilon')d\varepsilon d\varepsilon'. \\ 
	\end{split}
\end{equation*}

\subsection{Extension to Any Horizon $h \geq 2$}

The above results can be generalized for any $h \geq 2$ by using a recursive formula written in terms of the $h=1$ case. By a similar argument as in the Proof of Lemma 1, it can be shown that, for $h=2$:
\begin{equation*}
	\begin{split}
	& \left(\hat{g}_T^{(2)} - g^{(2)}\right)(y_t,\varepsilon_{t+1},\varepsilon_{t+2}) \\
\approx &  (\hat{g}_T-g)[g(y_t,\varepsilon_{t+1}),\varepsilon_{t+2}]+\left[\frac{\partial g^{(2)}}{\partial y'}(y_t,\varepsilon_{t+1},\varepsilon_{t+2})\right](\hat{g}_T-g)(y_t,\varepsilon_{t+1}). \\ 
	\end{split}
\end{equation*}
In general, the recursive formula is: 
\begin{equation*}
	\begin{split}
		& \left(\hat{g}_T^{(h)} - g^{(h)}\right)(y_t,\varepsilon_{t+1:t+h}) \\
		\approx &  (\hat{g}_T-g)[g^{(h-1)}(y_t,\varepsilon_{t+1:t+h-1}),\varepsilon_{t+h}]+\left[\frac{\partial g^{(h)}}{\partial y'}(y_t,\varepsilon_{t+h-1},\varepsilon_{t+h})\right](\hat{g}_T^{(h-1)}-g^{(h-1)})(y_t,\varepsilon_{t+1:t+h-1}), \\ 
	\end{split}
\end{equation*}
where $\varepsilon_{t+1:t+h}$ denotes the sequence of innovations $\varepsilon_{t+1},...,\varepsilon_{t+h}$. By repeated substitution, it can be shown that:
\begin{equation*}
	\begin{split}
	& \left(\hat{g}_T^{(h)} - g^{(h)}\right)(y_t,\varepsilon_{t+1:t+h}) \\
	\approx & (\hat{g}_T-g)[g^{(h-1)}(y_t,\varepsilon_{t+1:t+h-1}),\varepsilon_{t+h}] \\ 
	  & +  \sum_{i=3}^{h} \left[\prod_{j=i}^{h} \frac{\partial g^{(j)}}{\partial y'}(y_t,\varepsilon_{t+j-1},\varepsilon_{t+j})\right] (\hat{g}_T-g)[g^{(i-2)}(y_t,\varepsilon_{t+1:t+i-2}),\varepsilon_{t+i-1}] \\
	  & + \left[\prod_{j=2}^{h} \frac{\partial g^{(j)}}{\partial y'}(y_t,\varepsilon_{t+j-1},\varepsilon_{t+j})\right](\hat{g}_T-g)(y_t,\varepsilon_{t+1}). \\  
	\end{split}
\end{equation*}
Using this formula, we can now evaluate the asymptotic properties of $\widehat{IRF}_T(h,\delta) - IRF_T(h,\delta)$ for any $h \geq 2$. Let us demonstrate the case where $h=2$. We have that: 
\begin{equation*}
	\begin{split}
		& \widehat{IRF}_T(2,\delta) - IRF_T(2,\delta) \\
		= & \int\int\left\{[\hat{g}^{(2)}_T(y_t,\varepsilon_{t+1}+\delta,\varepsilon_{t+2})-g^{(2)}(y_t,\varepsilon_{t+1}+\delta,\varepsilon_{t+2})]\right.\\
		& \left.-[\hat{g}^{(2)}_T(y_t,\varepsilon_{t+1:t+2})-g^{(2)}(y_t,\varepsilon_{t+1:t+2})]\right\}\phi(\varepsilon_{t+1})\phi(\varepsilon_{t+2})d\varepsilon_{t+1}d\varepsilon_{t+2}. \\
		& = \int\int  \left\{(\hat{g}_T-g)[g(y_t,\varepsilon_{t+1}+\delta),\varepsilon_{t+2}]+\left[\frac{\partial g^{(2)}}{\partial y'}(y_t,\varepsilon_{t+1}+\delta,\varepsilon_{t+2})\right](\hat{g}_T-g)(y_t,\varepsilon_{t+1}+\delta)\right. \\ 
		& -\left. (\hat{g}_T-g)[g(y_t,\varepsilon_{t+1}),\varepsilon_{t+2}]+\left[\frac{\partial g^{(2)}}{\partial y'}(y_t,\varepsilon_{t+1},\varepsilon_{t+2})\right](\hat{g}_T-g)(y_t,\varepsilon_{t+1})\right\}\phi(\varepsilon_{t+1})\phi(\varepsilon_{t+2})d\varepsilon_{t+1}d\varepsilon_{t+2}. \\ 
	\end{split}
\end{equation*}


\subsection{Indirect Estimation}

If $S$ is large (i.e. $S=\infty$), then we have:
\begin{equation}\label{id_est}
	\widehat{\widehat{IRF}}_T(h,\delta) = \int \left\{\hat{m}_T^{(h-1)}[\hat{g}(y_t,\varepsilon_{t+1}+\delta)]-\hat{m}_T^{(h-1)}[\hat{g}(y_t,\varepsilon_{t+1})]\right\}\phi(\varepsilon_{t+1})d\varepsilon_{t+1}.
\end{equation}
For the case where $h=1$, it is important to remark that the equations \eqref{dd_irf} and \eqref{lp_irf} coincide with $m^{(0)}=Id$. Hence,there is no uncertainty on $\hat{m}_T^{(0)}$ and the estimators $\widehat{IRF}_T(1,\delta)$ and $\widehat{\widehat{IRF}}_T(1,\delta)$ are identical.On the other hand, when $h \geq 2$, we can show that the term in the integral expression of $\widehat{\widehat{IRF}}_T(h,\delta)$  has the asymptotic expansion:
\begin{equation*}
	\begin{split}
	& \hat{m}_T^{(h-1)}[\hat{g}(y_t,\varepsilon_{t+1}+\delta)]-m^{(h-1)}[\hat{g}(y_t,\varepsilon_{t+1})] \\ 
\approx & [\hat{m}_T^{(h-1)}-m^{(h-1)}][g(y_t,\varepsilon_{t+1})] + \frac{\partial m^{(h-1)}[g(y_t,\varepsilon_{t+1})]}{\partial y'}[\hat{g}_T(y_t,\varepsilon_{t+1})-g(y_t,\varepsilon_{t+1})], \\ 
	\end{split}
\end{equation*}
where $g(y_t,\varepsilon_{t+1})=Q[\Phi(\varepsilon_{t+1})|y_t]$ [from $m^{(h-1)}(y)=\mathbb{E}[y_{t+h-1}|y_t=y]$ [from equation \eqref{mh_lp}]. The first term in the expansion above captures the uncertainty of the Nadaraya-Watson estimator of the conditional expectation, and the second term the uncertainty when defining the Markov model. Thus, we have:
\begin{equation*}
	\begin{split}
			& h(T)\left[\widehat{\widehat{IRF}}_T(h,\delta)-IRF(h,\delta)\right] \\
	= & h(T)\int  \left\{[\hat{m}_T^{(h-1)}-m^{(h-1)}][g(y_t,\varepsilon_{t+1}+\delta)] + \frac{\partial m^{(h-1)}[g(y_t,\varepsilon_{t+1}+\delta)]}{\partial y'}[\hat{g}_T(y_t,\varepsilon_{t+1}+\delta)-g(y_t,\varepsilon_{t+1}+\delta)]\right. \\ 
		& \left.- [\hat{m}_T^{(h-1)}-m^{(h-1)}][g(y_t,\varepsilon_{t+1})] - \frac{\partial m^{(h-1)}[g(y_t,\varepsilon_{t+1})]}{\partial y'}[\hat{g}_T(y_t,\varepsilon_{t+1})-g(y_t,\varepsilon_{t+1})] \right\}\phi(\varepsilon_{t+1})d\varepsilon_{t+1}.\\
	\end{split}
\end{equation*}

The expansion above shows the main difficulties when looking for the asymptotic distribution of the nonparametric indirect estimator of the IRF. Within the integral, we have two types of terms that involve either functions $\hat{g}_T(y_t,\varepsilon_{t+1})-g(y_t,\varepsilon_{t+1})$, say, and $[\hat{m}_T^{(h-1)}-m^{(h-1)}][g(y_t,\varepsilon_{t+1})]$, say. From Online Appendix C.1, Lemma 1, we know that the first type of terms are converging at a rate of $h(T)=\sqrt{Tb_T}$ (in the one dimensional case) [see online Appendix C] and that in their expansion the rate is coming from the conditioning value $y_t$, not from the argument $\varepsilon_{t+1}$. Thus, we expect after integration with respect to $\varepsilon$ to still have the same rate $h(T)=\sqrt{Tb_T}$ and the asymptotic normality. \\

Let us now consider the second-type term. The difference $\left[\hat{m}_T^{(h-1)}-m^{(h-1)}\right](y)$ based on the Nadaraya-Watson estimator is still consistent, at a speed of $h(T)=\sqrt{Tb_T}$ and asymptotically normal. However, when $y$ is replaced by $g(y_t,\varepsilon)$ and reintegrated with respect to $\varepsilon$, this introduces a smoothing with respect to the conditioning value. Loosely speaking, the conditional expectation to be estimated is replaced by a kind of unconditional expectation. In other words, after reintegration, we expect these terms to converge at a parametric rate [see Ait-Sahalia (1993) for the careful analysis of the effect of integration on functional estimators and the use of Hadamard derivatives]. To summarize, the terms of this second type are negligible with respect to the terms of the first type in the expansion above. Therefore, we have: 
\begin{equation}
\begin{split}
		& h(T)\left[\widehat{\widehat{IRF}}_T(h,\delta)-IRF(h,\delta)\right] \\ 
		\approx & \int\left\{\frac{\partial m^{(h-1)}[g(y_t,\varepsilon+\delta)]}{\partial y'}h(T)[\hat{g}_T(y_t,\varepsilon+\delta)-g(y_t,\varepsilon+\delta)]\right.  \\
		& \left. -\frac{\partial m^{(h-1)}[g(y_t,\varepsilon)]}{\partial y'}h(T)[\hat{g}_T(y_t,\varepsilon)-g(y_t,\varepsilon)] \right\}\phi(\varepsilon)d\varepsilon. \\
\end{split}
\end{equation}

\subsection{Comparison of Accuracies of the Direct and Indirect Approaches}

We can now compare the expansions obtained from the estimated IRF's in the direct and indirect nonparametric approaches. Their expansions differ by the multiplicative terms that are:
\begin{equation*}
	\prod_{j=2}^{h}\frac{\partial g^{(h)}}{\partial y'}(y_t,\varepsilon_{t+j-1},\varepsilon_{t+j}), \ \text{reintegrated w.r.t.} \ \varepsilon_{t+1},...,\varepsilon_{t+h},
\end{equation*}
for the direct approach, and 
\begin{equation*}
	\frac{\partial m^{(h-1)}}{\partial y'}(y_t,\varepsilon_{t+1}), \ \text{reintegrated w.r.t.} \ \varepsilon_{t+1},
\end{equation*}
for the indirect approach. It is easy to check that the effect of these multiplicative terms after reintegration are the same as a consequence of the formula of iterated expectations.

\subsection{Asymptotic Expansion of the Estimated IRF in the Semi-Parametric Model}

\textbf{(i) The estimation error on the IRF} \\

The true IRF at horizon $h=1$ is given in \eqref{sp_irf_true}:
\begin{equation}
\begin{split}
	&	IRF(1,\delta;\beta,(Q_{i})) \\
	=& \int ... \int\left[ g(y_t,\text{vec}(Q_i\circ\Phi(\varepsilon_{i}+\delta));\beta)-g(y_t,\text{vec}(Q_i\circ\Phi(\varepsilon_{i}));\beta)\right]\phi(\varepsilon_1)...\phi(\varepsilon_n)d\varepsilon_1...d\varepsilon_n\\
\end{split}
\end{equation}
We have: 
\begin{equation}\label{irfexp}
\begin{split}
	&IRF(1,\delta;\hat{\beta},(\hat{Q}_{i}))-	IRF(1,\delta;\beta,(Q_{i}))\\
	= & \int ... \int \left[ g(y_t,\text{vec}(\hat{Q_i}\circ\Phi(\varepsilon_{i}+\delta));\hat{\beta})-g(y_t,\text{vec}(Q_i\circ\Phi(\varepsilon_{i}+\delta));\beta)\right]\\ 
	& - \left[ g(y_t,\text{vec}(\hat{Q_i}\circ\Phi(\varepsilon_{i}));\hat{\beta})-g(y_t,\text{vec}(Q_i\circ\Phi(\varepsilon_{i}));\beta)\right]\phi(\varepsilon_1)...\phi(\varepsilon_n)d\varepsilon_1...d\varepsilon_n\\
\end{split}
\end{equation}

\textbf{(ii) Taylor expansion}\\

Then, if the number of observations is large, the estimators $\hat{\beta}$, $\hat{Q}_i$ converge to their true values $\beta$, $Q_i$, and we can consider a Taylor expansion of the term within the integral. We get:
\begin{equation}\label{te_a}
	\begin{split}
&	g(y_t,\text{vec}(\hat{Q_i}\circ\Phi(\varepsilon_{i}));\hat{\beta})-g(y_t,\text{vec}(Q_i\circ\Phi(\varepsilon_{i}));\beta)\\
\approx & D_u g(y_t,\text{vec}(Q_i\circ\Phi(\varepsilon_{i}));\beta)\left[\text{vec}\left((\hat{Q}_i-Q_i)\circ\Phi(\varepsilon_i)\right)\right] + D_\beta g(y_t,\text{vec}(Q_i\circ\Phi(\varepsilon_{i}));\beta)\left[\hat{\beta}-\beta\right], \\
	\end{split}
\end{equation}
where $D_u g(\cdot)$ is a Jacobian matrix with elements $\left[D_u g_i(\cdot)\right]_{i,j} = \frac{\partial g(\cdot)}{\partial u_j}\vert_{u_j=Q_j\circ\Phi(\varepsilon_j)}$, and $D_\beta g(\cdot)$ is a gradient vector with elements $\left[D_\beta g(\cdot)\right]_i=\frac{\partial g_i(\cdot)}{\partial \beta}$. Taking into account these expansions, we deduce from \eqref{irfexp} and \eqref{te_a}: 
\begin{equation}\label{irfex}
\begin{split}
& IRF(1,\delta;\hat{\beta},(\hat{Q}_{i}))-	IRF(1,\delta;\beta,(Q_{i}))\\
=& \int...\int \left\{D_u g(y_t,\text{vec}(Q_i\circ\Phi(\varepsilon_{i}+\delta));\beta)\left[\text{vec}\left((\hat{Q}_i-Q_i)\circ\Phi(\varepsilon_i+\delta)\right)\right]\right.\\
& -\left.D_u g(y_t,\text{vec}(Q_i\circ\Phi(\varepsilon_{i}));\beta)\left[\text{vec}\left((\hat{Q}_i-Q_i)\circ\Phi(\varepsilon_i)\right)\right]\right.\\
& + \left.\left[D_\beta g(y_t,\text{vec}(Q_i\circ\Phi(\varepsilon_{i}+\delta));\beta)- D_\beta g(y_t,\text{vec}(Q_i\circ\Phi(\varepsilon_{i}));\beta)\right]
\left[\hat{\beta}-\beta\right] \right\}\phi(\varepsilon_1)...\phi(\varepsilon_n)d\varepsilon_1...d\varepsilon_n\\ 
& + o_p(1),\\
\end{split}
\end{equation}
where $o_p(1)$ is negligible in probability. \\ 

\textbf{(iii) The asymptotic behaviour of the estimators}\\

The expansions above show the effect of estimation errors on $\beta$ and $Q_i$. By assumption \eqref{abeta} we have:
\begin{equation}
	\sqrt{T}(\hat{\beta} - \beta) \xrightarrow[]{d} N(0,V(\beta)),
\end{equation}
or equivalently $\sqrt{T}(\hat{\beta}_T - \beta_0)\xrightarrow[]{d} V(\beta)^{\frac{1}{2}}Z^*$, where $Z^*$ is a standard normal random variable. Moreover, analogous to the proof of Lemma 1 (see Online Appendix B), we have: 
\begin{equation}\label{qia1}
	\sqrt{T}(\hat{Q}_{i}-Q_{i})\circ[\Phi(\varepsilon_i)] \approx = - \frac{1}{f_i(Q_i\circ(\Phi(\varepsilon_i)))}\sqrt{T}(\hat{F}_i-F_i)\circ\left[Q_i(\Phi(\varepsilon_i))\right],
\end{equation}
for $i=1,...,n$. Then, by using Donsker's Theorem:
\begin{equation}
	\sqrt{T}\left[\hat{F}_i-F_i\right](Q_i(\Phi(\varepsilon_i)) \xrightarrow[]{d} BB_i(Q_i(\Phi(\varepsilon_i)), \ i=1,...,n,
\end{equation}
where the $BB_i$, $i=1,...,n$ are independent Brownian bridges (since they are obtained from the independent components of $\varepsilon$), and $\xrightarrow[]{d}$ denotes the convergence in distribution of the processes (indexed by $\varepsilon_i$). Thus, we deduce that:
\begin{equation}\label{qia2}
	\sqrt{T}(\hat{Q}_{i}-Q_{i})\circ[\Phi(\varepsilon_i)]  \xrightarrow[]{d} \frac{1}{f_i(Q_i\circ(\Phi(\varepsilon_i)))}BB_i(Q_i(\Phi(\varepsilon_i)), \ i=1,...,n.
\end{equation}
Note that in these asymptotic results the limiting variable $Z^*$ and processes $BB_i$ are independent of the observations $(y_t)$. Furthermore, the asymptotic results in \eqref{qia1}-\eqref{qia2} do not take into account the fact that the estimated quantiles $\hat{Q}_i$ are computed from the residuals $\hat{u}_t$ and not from the error themselves (see the discussion in point (v) below). Finally, these asymptotic behaviours are deduced from first order conditions for $\hat{\beta}-\beta$ and $\hat{Q}_i(\varepsilon_i)-Q_i(\varepsilon_i)$ that involve different scores on which the Central Limit Theorem (CLT) and Functional Central Limit Theorem (FCLT) are applied to obtain ``asymptotic normality" [see Online Appendix D]. When they are considered jointly, we get a joint normality in which the variable $Z^*$ and the Brownian Bridges can be correlated. \\

\textbf{(iv) Asymptotic distribution of the estimated IRF} \\ 

Next, we introduce the asymptotics for  $\sqrt{T}(IRF(1,\delta;\hat{\beta},(\hat{Q}_{i}))-	IRF(1,\delta;\beta,(Q_{i})))$. From \eqref{irfex} we get:
\begin{equation}\label{finalirf}
\begin{split}
	& \sqrt{T}\left[IRF(1,\delta;\hat{\beta},(\hat{Q}_{i}))-	IRF(1,\delta;\beta,(Q_{i}))\right]\\
	= & \int...\int \left\{D_u g(y_t,\text{vec}(Q_i\circ\Phi(\varepsilon_{i}+\delta));\beta)\left[\text{vec}\left(- \frac{1}{f_i(Q_i\circ(\Phi(\varepsilon_i)))}\sqrt{T}(\hat{F}_i-F_i)\circ\left[Q_i(\Phi(\varepsilon_i+\delta))\right]\right)\right]\right.\\
	& -\left.D_u g(y_t,\text{vec}(Q_i\circ\Phi(\varepsilon_{i}));\beta)\left[\text{vec}\left(- \frac{1}{f_i(Q_i\circ(\Phi(\varepsilon_i)))}\sqrt{T}(\hat{F}_i-F_i)\circ\left[Q_i(\Phi(\varepsilon_i))\right]\right)\right]\right.\\
	& + \left.\left[D_\beta g(y_t,\text{vec}(Q_i\circ\Phi(\varepsilon_{i}+\delta));\beta)- D_\beta g(y_t,\text{vec}(Q_i\circ\Phi(\varepsilon_{i}));\beta)\right]
	\left[\hat{\beta}-\beta\right] \right\}\phi(\varepsilon_1)...\phi(\varepsilon_n)d\varepsilon_1...d\varepsilon_n + o_p(1)\\ 
\end{split}
\end{equation}

In particular, we have: 
\begin{equation}
\begin{split}
	&  \sqrt{T}\left[IRF(1,\delta;\hat{\beta},(\hat{Q}_{i}))-	IRF(1,\delta;\beta,(Q_{i}))\right]\\
\stackrel{d}{\approx} & \int...\int \left\{D_u g(y_t,\text{vec}(Q_i\circ\Phi(\varepsilon_{i}+\delta));\beta)\left[\text{vec}\left(- \frac{1}{f_i(Q_i\circ(\Phi(\varepsilon_i)))}BB_i\left[Q_i(\Phi(\varepsilon_i+\delta))\right]\right)\right]\right.\\
& -\left.D_u g(y_t,\text{vec}(Q_i\circ\Phi(\varepsilon_{i}));\beta)\left[\text{vec}\left(- \frac{1}{f_i(Q_i\circ(\Phi(\varepsilon_i)))}BB_i\left[Q_i(\Phi(\varepsilon_i))\right]\right)\right]\right.\\
& + \left.\left[D_\beta g(y_t,\text{vec}(Q_i\circ\Phi(\varepsilon_{i}+\delta));\beta)- D_\beta g(y_t,\text{vec}(Q_i\circ\Phi(\varepsilon_{i}));\beta)\right]
\left[\hat{\beta}-\beta\right] \right\}\phi(\varepsilon_1)...\phi(\varepsilon_n)d\varepsilon_1...d\varepsilon_n, \\ \\	
\end{split}
\end{equation}
where $\stackrel{d}{\approx}$ denotes the asymptotic equivalence in distribution.\\

 Then we deduce that $\sqrt{T}\left[(IRF(1,\delta;\hat{\beta},(\hat{Q}_{i}))-	IRF(1,\delta;\beta,(Q_{i})))\right]$ is asymptotically normal with zero mean since we get a linear functional of Gaussian processes (including the degenerate $Z^*$). Note that the equivalence (A.12) is currently writtten for given environment $(y_t)$ and given shock of magnitude $\delta$. With stronger regularity conditions, we can consider (A.11) as valid in terms of processes indexed by $y_t$ and $\delta$, that is a functional limit theorem. The asymptotic variance could be completed explicitly along the lines of the univariate case in Appendix A.1. \\

\textbf{(v) The two-step estimation effect}\\

As mentioned above, we have not yet taken into account the fact that the estimator $\hat{Q}_i$ is completed from residuals instead of being computed with errors. If we distinguish these two estimators and denote $\hat{\hat{Q}}_i$, the estimator computed from residuals, then the estimated IRF is obtained with $\hat{\hat{Q}}_i$. However, we will get another expansion for $\hat{\hat{Q}}_i-{Q}_i$. Then the effect of $Z^*$ in expansion \eqref{finalirf} will be modified [see online Appendix D.4.].

\section*{Online Appendix B Proof of Lemma 1}

We have the following expansion for nonparametric estimators $\hat{Q}_T, \hat{F}_T$ of the quantile and cdf: 
\begin{equation*}
	\begin{split}
		& y = \hat{Q}_T[\hat{F}_T(y)] \\ 
		\iff &  \hat{Q}_T[\hat{F}_T(y)]-Q(F(y)) = 0 \\
		\iff &  \{\hat{Q}_T[\hat{F}_T(y)]-\hat{Q}[F(y)]\}+\{\hat{Q}_T[F(y)]-Q(F(y))\}=0 \\
		\iff &   \frac{d\hat{Q}_T}{d\alpha}[F(y)](\hat{F}_T(y)-F(y))+\{\hat{Q}_T[F(y)]-Q(F(y))\}\approx0\\ 
		& \text{ (By a Taylor expansion of $\hat{Q}_T[\hat{F}_T(y)]$)}\\
		\iff & \{\hat{Q}_T[F(y)]-Q(F(y))\} =  - \frac{d\hat{Q}_T}{d\alpha}[F(y)](\hat{F}_T(y)-F(y)) \\ 
		\iff & (\hat{Q}_T-Q)[F(y)] \approxeq \frac{-1}{f(Q(y))}(\hat{F}_T-F)(y).\\ 
	\end{split}
\end{equation*}
Substituting $F(y)=\Phi(\varepsilon_t)$ and conditioning on $y_{t}$: 
\begin{equation*}
	(\hat{Q}_T-Q)[\Phi(\varepsilon)|y_{t}]\approxeq - \frac{1}{f(Q(\Phi(\varepsilon)|y_{t}))} (\hat{F}_T-F)[Q(\Phi(\varepsilon_t)|y_{t})|y_{t}],
\end{equation*}
we have the desired result. \qed \\

\section*{Online Appendix C: Asymptotic Results for Nonparametric Functional Estimators}

We consider several kernel-based functional estimators for our nonparametric approach and discuss some of their asymptotic properties here. The definitions of the different kernel based functional estimators are standard [see Nadaraya (1964) for regression, Loadler (1996), Silverman (2018) for density estimation, Xiang (1996), Yu and Jones (1998) for quantile regression, Falk (1983) for the kernel quantile estimator]. Their asymptotic properties are valid for stochastic processes [Bosq (2012)].  \\

\textbf{(i) For the density of $y_t$:} \\ 

Suppose $y_t$ is one dimensional with true density $f(y)$. We denote its kernel density estimator as: 

\begin{equation}\label{kde}
	\hat{f}(y)= \frac{1}{Tb_T} \sum_{t=1}^TK\left(\frac{y_t-y}{b_T}\right),
\end{equation}

where  $b_T$ denotes the bandwidth and $K$ is the kernel function which is non-negative and satisfies two properties: [1] The integral of $K(\cdot)$ over its support is equal to unity, that is, $\int K(u)du = 1$. [2] $K$ is a symmetrical function such that $K(-x)  = K(x)$ for all values of $x$ in its support. \\

Under standard regularity conditions (see online Appendix D.3), this estimator is consistent and asymptotically normal:

\begin{equation*}
	\sqrt{Tb_T}\left\{\hat{f}(y)-f(y)\right\} \rightarrow N\left(0,f(y)\int K^2(u)du\right),
\end{equation*}
if $T \rightarrow \infty$, $b_T \rightarrow 0$ and $Tb_T^{5/3} \rightarrow \infty$. \\

\textbf{(ii) For the conditional cdf of $y_t$ given $y_{t-1}$:}\\

This is given by the expression: 

\begin{equation}\label{ccdf}
	\hat{F}(\mathfrak{z}|y) = \sum_{t=2}^T \mathbbm{1}_{y_t<\mathfrak{z}} K\left(\frac{y_t-y}{b_T}\right)/\sum_{t=2}^TK\left(\frac{y_t-y}{b_T}\right).
\end{equation}

Under standard regularity conditions [see online Appendix D.3], this estimator is consistent and asymptotically normal. 
\begin{equation*}
	\sqrt{Tb_T}\left[\hat{F}(\mathfrak{z}|y)-F(\mathfrak{z}|y)\right]\rightarrow N\left[0,\frac{\int K^2(u)duF(\mathfrak{z}|y)(1-F(\mathfrak{z}|y))}{f(y)}\right],
\end{equation*}
if $T \rightarrow \infty$, $b_T \rightarrow 0$ and $Tb_T^{5/3} \rightarrow \infty$. \\

\textbf{(iii) For the conditional quantile of $y_t$ given $y_{t-1}=y$:}\\

The kernel estimator is defined by: 

\begin{equation}\label{cqe}
	\hat{Q}(\alpha|y) = \argmin_{q} \sum^{T}_{t=1}K\left(\frac{y_{t-1}-y}{b_T}\right)\{\alpha(y_t-q)^+ + (1-\alpha)(y_t-q)^-\}.
\end{equation}
 Thus it is obtained as a kernel M. estimator\footnote{See also Falk (1985) for an alternative kernel quantile estimator based on a kernel average of order statistics.}. This approach is largely used in financial applications for the nonparametric estimation of a conditional Value-at-Risk [Gourieroux and Jasiak (2010)]. Under standard regularity conditions [see online Appendix D.3], this estimator is consistent and asymptotically normal:
\begin{equation*}
	\sqrt{Tb_T}\left[\hat{Q}(\alpha|y)-Q(\alpha|y)\right]\rightarrow N\left[0,\frac{\int K^2(u)du \ \alpha \ (1-\alpha)}{f^2(Q(\alpha|y)|y)f(y)}\right],
\end{equation*}
if $T \rightarrow \infty$, $b_T \rightarrow 0$ and $Tb_T^{5/3} \rightarrow \infty$. \\

\textbf{(iv) For the conditional expectation of $y_{t+h}$ given $y_t$: }\\

Consider the Nadaraya-Watson estimator, given by: 
\begin{equation}\label{nwe}
	\begin{split}
		\hat{m}_{h,T}(y) & = \argmin_{m} \sum^{T-h}_{t=1}K\left(\frac{y_t-y}{b_T}\right)(y_{t+h}-m)^2\\
		& = \sum_{t=1}^{T-h}\left[K\left(\frac{y_t-y}{b_T}\right)y_{t+h}\right]/\sum_{t=1}^{T-h}K\left(\frac{y_t-y}{b_T}\right),\\
	\end{split}
\end{equation}
which is used in equation (4.4) for the local projection. Under mild regularity conditions, this estimator is asymptotically normal: 
\begin{equation}
	\sqrt{Tb_T}\left[\hat{m}_{h,T}(y)-m(y)\right] \rightarrow N\left(0,\frac{\sigma^2(y)\int_{-\infty}^{+\infty}K^2(u)du}{f(y)}\right),
\end{equation}
if $T\rightarrow \infty$ and $b_t = O(T^{1/3})$.\\

\section*{Online Appendix D: Regularity Conditions}

We provide in this appendix a list of regularity conditions for the semiparametric estimation of the IRF. They are obtained from Gourieroux and Jasiak (2023) for the first step GCov estimators, Falk (1985), and Ait-Sahalia (1983) for integrals of functional estimators.

\subsection*{D.1 Model}

\begin{enumerate}[label=a.\arabic*, start=1, series=a]
	\item The semiparametric model is given by $y_t=g(y_{t-1},u_t;\beta)$ as in (5.4), where the components of the errors are $u_t = F_{i}^{-1}\circ \Phi(\varepsilon_{i,t})$ and $\varepsilon_{t}$ is $IIN(0,Id)$.
	\item The function $u_t \rightarrow g(y_{t-1},u_t,\beta)$ is invertible and then the model can be written equivalently as $G(y_t,y_{t-1},\beta)=u_t$ as in (5.5).
	\item Functions g and $G$ are known. $\beta$ and $(F_i)$ are parameters to be estimated. 
	\item The distributions $F_i$, $i=1,...,n$, are continuous on $\mathbb{R}$, with strictly positive continuous densities. 
	\item The semiparametric model is well-specified with $\beta_0$, $(F_{i,0})$ as the true values of the parameters. 
	\item There exists, for each $\beta$, $(F_{i})$ in a neighborhood of the true values, a unique strictly stationary solution $(y_t)$ to (5.5) or (5.6). The observations $y_1,...,y_t$ correspond to such a strictly stationary solution. 
	\item $u_t$ is independent of $y_{t-1}$, $y_{t-2}$,...
	\item The process $(y_t)$ admits a nonlinear (infinite) moving average representation with respect to $u_t,u_{t-1},...$
\end{enumerate}

Assumptions a.6 and a.8 are high level assumptions that have to be checked case by case. For instance, in the DAR example of Section 4.3, one can verify that the parameters satisfy the stationarity condition $\mathbb{E}\left[|\phi+\sqrt{\beta}\varepsilon_t\right]<0$ for $\phi=0.5$ and $\beta=0.5$. 

%
\subsection*{D.2 First-Step Estimation}

The parameter $\beta$ can be estimated semiparametrically in various ways. Let us consider a basic GCov estimator based on the sample analogue of the theoretical covariance restrictions: 
\begin{equation}\label{GCOV}
	Cov(a_k\left[g_i(Y_{t},Y_{t-1},\beta)\right],a_\ell\left[g_j(Y_t,Y_{t-1},\beta)\right])=0, \ k,\ell=1,...,K, \ k\neq \ell.
\end{equation}
\begin{enumerate}[label=a.\arabic*, resume=a]
	\item The parameter set $B$ for $\beta$ is compact, with non-empty interior \r{B}.
	\item For the true values of the parameters $\beta$, $(F_i)$, the equation \eqref{GCOV} has a unique solution $\beta_0$ in \r{B} (This is the asymptotic identification condition for the GCov estimator). 
	\item The transformations $a_k\left[g_i(Y_{t},Y_{t-1},\beta)\right]$, are integrable at fourth-order under the true model.
	\item The transformations $a_k\left[g_i(Y_{t},Y_{t-1},\beta),k=1,...,K\right]$ are linearly dependent, for any $i$. 
	\item The functions $a_k$ $k=1,...,K$, are twice continuously differentiable. 
	\item The functions $\beta \rightarrow g_i(Y_t,Y_{t-1},\beta)$  are twice continuously differentiable on \r{B}. 
	\item $\Gamma = \left[\mathbb{V}\left(a_k\left[g_i(Y_{t},Y_{t-1},\beta)\right]\right)\right]^{-1}$ is continuous in $\beta \in$ \r{B} (note that it exists by a.11).
	\item Additional technical conditions in order to apply the Lindeberg-Feller theorem for asymptotic normality. 	
\end{enumerate}
Under such conditions, we can prove the existence, consistency and asymptotic normality of the GCov estimator. 

\subsection*{D.3 Functional Estimators}

Now we need to include the conditions for functional estimation of the $F_i$ based on the true errors $u_{i,0} = G(y_t,y_{t-1},\beta_0)$. Note that these true errors $u_{i,0}$ depend on $\beta_0$, which is unknown. Therefore, the associated ``estimators" using these true values $\beta_0$ are infeasible. \\

First, for the kernel function we assume: 
\begin{enumerate}[label=a.\arabic*, resume=a]
\item The kernel function $K$ is such that $\int K(u) du = 1$ and symmetrical, that is, $K(-u) = K(u)$ for all $u$ in the support of $K$. 
\item The kernel function is bounded, that is, $K(u)=0$ for $|u|> C$, $C < \infty$. 
\end{enumerate}
We also impose a stationarity condition:
\begin{enumerate}[label=a.\arabic*, resume=a]
		\item $(Y_t)$ is a $\alpha-$mixing process such that $|\alpha(\ell)| \geq C\ell^{-\beta}$ for $C>0$ and $\beta>2$, where $\alpha(\ell)$ is the $\alpha-$mixing coefficient at lag $\ell$. 
\end{enumerate}
For the kernel density estimator in \eqref{kde}:
\begin{enumerate}[label=a.\arabic*, resume=a]
	\item Let $g_\ell(y,y')$ be the joint density between $Y_t$ and $Y_{t-1}$. It satisfies uniform boundedness, that is $||g_\ell||_\infty < \infty$. 
	\item The true density $f(y)$ has a continuous second-order derivative at the point $y$.
	\item $Tb_T^{5/3} \rightarrow \infty$, where $b_T$ denotes the bandwidth. 
\end{enumerate}
For the conditional c.d.f. estimator in \eqref{ccdf}:
\begin{enumerate}[label=a.\arabic*, resume=a]
	\item The conditional c.d.f. $Y_t$ given $Y_{t-1}$, $G_{Y_t|Y_{t-1}}(y|y')$, is positive and continuous at $y$. 
	\item $G_{Y_t|Y_{t-1}}(y|y')$ is twice differentiable. 
\end{enumerate}
For the conditional quantile estimator in \eqref{cqe}:
\begin{enumerate}[label=a.\arabic*, resume=a]
	\item The conditional density of $Y_t$ given $Y_{t-1}$, $g_{Y_t|Y_{t-1}}(y|y')$, is positive at the given quantile. 
	\item The conditional density $g_{Y_t|Y_{t-1}}(y|y')$ twice differentiable around the neighbourhood of the given quantile.
\end{enumerate}
For the Nadaraya-Watson estimator in \eqref{nwe}:
\begin{enumerate}[label=a.\arabic*, resume=a]
	\item The function $m$ is continuous and differentiable at $y$. 
	\item The marginal density of $(Y_t)$ is continuous. 
	\item The joint density between $Y_t$ and $Y_{t-1}$ is bounded such that $g_h(y,y')< A_1 < \infty$. 
		\item The conditional density between $Y_t$ and $Y_{t-1}$ is bounded such that $g_h(y|y')< A_2 < \infty$. 
	\item The conditional variance $\mathbb{V}\left[Y_{t+h}|Y_t=y\right]$ is continuous at $y$. 
\end{enumerate}
Finally, tail conditions are also needed to allow for a functional limit theorem for a kernel estimator of $Q$ and the integrations appearing in the last step of the estimation for the IRF. Such a condition can be [Viallon (2007)]:
\begin{equation*}
	\sup_{u \in (-\infty,+\infty)} 
		\left\{ 
		F(u)\,[1-F(u)] \, 
		\left| \frac{df(u)}{du} \right| 
		\,\bigg/\, f^2(u) 
		\right\} 
		\leq \gamma, 
		\quad\text{with } \gamma > 0,
\end{equation*}
or equivalently,
\begin{equation*}
	\sup_{t \in (0,1)} 
\left\{ 
t(1-t) \, 
\left| \frac{d f\!\left[ Q(t) \right]}{du} \right| 
\,\bigg/\, f^2(Q(t)) 
\right\} 
\leq \gamma, 
\quad\text{with } \gamma > 0,
\end{equation*}
written for any distribution $F_i$,$f_i$, $Q_i$ in the case of independent sources. 

\subsection*{D.4 Residuals and Feasible Functional Estimation}

The feasible nonparametric estimators $\hat{Q}_i$ of $Q_i$ are computed from the residuals of the first-step estimation, not from the true errors themselves. Thus, additional assumptions are required to manage the influence of these first step estimation errors on the nonparametric estimator of $F_{i,0}$ and $Q_{i,0}$ to ensure consistency and to derive the appropriate first-order expansion [see equation (A.11) and the discussion in Appendix A.5.(v).]. The main ones are the following:
\begin{enumerate}[label=a.\arabic*, resume=a]
	\item $\lim_{T\rightarrow\infty} \sup_{t=1,...,T} \mathbb{E}\left[\hat{u}_{t}-u_t\right]^2=0$.
\end{enumerate}
Since we have:
\begin{equation*}
	u_t = G(y_{t},y_{t-1},\hat{\beta}_T) - G(y_t,y_{t-1},\beta_0) \approx \frac{\partial G}{\partial \beta'} [y_t,y_{t-1},\beta_0]\left(\hat{\beta}_T - \beta_0\right),
\end{equation*}
we see that:
\begin{equation}
	\sqrt{T}\left(\hat{u}_T-u_t\right) =  \frac{\partial G}{\partial \beta'} [y_t,y_{t-1},\beta_0]\sqrt{T}\left(\hat{\beta}_T - \beta_0\right)\\
	\end{equation}

Therefore, the sample cdf based on residuals is linked to the sample cdf based on true errors. This expansion involves two types of sums: the standard definition of the c.d.f: $F_i(u) \sim \hat{\hat{F_i}}(u)= \frac{1}{T}\sum_{t=1}^T\textbf{1}_{u_{i,t}\leq u}$, and the one appearing in the first-order expansion of the GCov estimator:
\begin{equation*}
	\sqrt{T}(\hat{\beta}-\beta_0) = \frac{1}{\sqrt{T}} \sum_{t=1}^T \mu_{T,t}(\beta_0) + o_p(1),
\end{equation*}
where $\mu_{T,t}(\beta_0)$ is a martingale difference sequence. The next condition ensures that it is possible to manage jointly these two expansions. Among technical conditions for the uniformity of the negligible terms $o_p(1)$, the important one is:
\begin{enumerate}[label=a.\arabic*, resume=a]
	\item The Lindeberg-Feller conditions for joint normality are assumed satisfied fro the multivariate martingale $\left[\textbf{1}_{u_{it} < u_m}- F_t(u_m), \; i = 1,..., n,\; m = 1,..., M,\; \mu_{T,t}(\beta_0)\right]$.
	\item Joint tightness conditions introduced to get a joint FCLT, when the set of values $u_m$ is chosen in an increasing grid such that $\lim_{M\rightarrow\infty} \bigcup_m^M \left[u_1,...,u_m\right]$ is dense in $[0,1]$. 
\end{enumerate}
Then we will derive an asymptotic expansion of $\sqrt{T}(F^*_i-F_{0,i})$ in terms of the Brownian bridge $BB_i$, and of $\sqrt{T}(\hat{\beta}-\beta)$, that will introduce an additional term in $\sqrt{T}(\hat{\beta}-\beta)$ for the expansion (A.11). \\

Modfiying the result in Appendix B we get: 
\begin{equation}
\begin{split}
&\sqrt{T} \, (\hat Q_T - Q)[F(y)]\\
= & -\frac{1}{f(Q(y))}\sqrt{T}\,(\hat{F}_T - F)(y)
+ \frac{f(y)}{f(Q(y))}
\left( \frac{1}{T}\sum_{t=1}^{T} 
\frac{\partial G}{\partial \beta'}[y_t, y_{t-1}; \beta_0] \right)'
\sqrt{T}(\hat\beta_T - \beta_0),\\
\end{split}
\end{equation}
since $\textbf{1}_{\hat{u}_T \leq y}-\textbf{1}_{u_T \leq y}\approx -f(y)(\hat{u}_t-u_T)$. By averaging over $t=1,...,T$, we get $-\frac{1}{T}\sum_{i=t}^Tf(y)\hat{u}_t-u_t$. Then this modified result from appendix B can be applied in a straightforward way to (A.11) to get a similar result. 

\subsection*{D.5 Expansion of the IRF}

When coming to the asymptotic expansion of the IRF, we need additional conditions in order to compute the limit in distribution and the integrals with respect to $\varepsilon_1,...,\varepsilon_n$ with (A.11) and (A.12). Such technical conditions can be found for instance in Ait Sahalia (1983). They demand the introduction of Hadamard derivatives in order to justify the appropriate Taylor expansions with respect to the functions in the Skorokhod space, $\mathcal{D}[0,1]$. We refer the reader to this paper in order to get an idea of these assumptions. \\

\textbf{References:}\\

Adjoudj, L. and A., Tatachak. (2019). Some Properties of a Kernel Conditional Quantile Estimator for Associated Data. \textit{ISTE, Biostatistics and Health}. Issue 1. \\

Bosq, D. (2012). Nonparametric Statistics for Stochastic Processes: Estimation and Prediction. Vol. 110. Springer Science and Business Media. \\

Gourieroux, C., and J., Jasiak. (2010a). Local Likelihood Density Estimation and Value‐at‐Risk. \textit{Journal of Probability and Statistics}, 2010(1), 754-851. \\

Gourieroux, C., and J., Jasiak. (2010b). Value-at-Risk. \textit{Handbook of Financial Econometrics: Tools and Techniques}. p 553-615. \\

Loadler, C. (1996). Local Likelihood Density Estimation. \textit{The Annals of Statistics}, 24(4), 1602-1618.\\

Nadaraya, E. (1964). On Estimating Regression. \textit{Theory of Probability and Its Applications}, 10, 186-190.\\

Silverman, B. (2018). Density Estimation for Statistics and Data Analysis. Routledge.\\

Viallon, B. (2007). \textit{Functional Limit Law for the Increments of the Quantile Process with Applications}. Electronic Journal of Statistics, 1, 496-518.\\

Xiang, X. (1996). A Kernel Estimator of a Conditional Quantile. \textit{Journal of Multivariate Analysis}, 59(2), 206-216.\\

Yu, K., and M., Jones. (1998). Local Linear Quantile Regression. \textit{Journal of the American Statistical Association}, 93(441), 228-237.\\

\end{appendices}

\end{document}